\documentclass[
aps,%
pre,%
reprint,%
showpacs,%
superscriptaddress,%
groupedaddress,%
floatfix,%
showkeys,%
nofootinbib%
]{revtex4-1}

\usepackage{float}
\usepackage{verbatim}
\usepackage{amsmath,amsfonts,amsthm,amssymb,times,mathrsfs,txfonts}
\usepackage[hidelinks]{hyperref}
\usepackage{bm}
\usepackage{multirow}

\newcommand{\field}[1]{\mathbb{#1}}
\newcommand{\fs}[1]{\mathsf{#1}}

\DeclareMathOperator*{\supp}{supp}
\DeclareMathOperator*{\esssup}{ess\,sup}
\DeclareMathOperator{\Res}{Res}
\DeclareMathOperator{\diag}{diag}

\newcommand{\tp}{\intercal}

\newcommand{\ovl}[1]{\overline{#1}}
\newcommand{\bigO}[1]{\mathop{O}(#1)}

\let\Re\relax
\DeclareMathOperator{\Re}{Re}
\let\Im\relax
\DeclareMathOperator{\Im}{Im}
\newcommand{\vv}[1]{\boldsymbol{#1}}
\newcommand{\vs}[1]{\boldsymbol{#1}}

\newcommand{\OP}[1]{\mathscr{#1}}

\DeclareMathOperator{\fourier}{\mathscr{F}}

\DeclareMathOperator{\sech}{sech}

\DeclareMathOperator{\sinc}{sinc}

\newtheorem{theorem}{Theorem}[section]
\newtheorem{prop}[theorem]{Proposition}
\newtheorem{corr}[theorem]{Corollary}
\newtheorem{lemma}[theorem]{Lemma}
\newtheorem{defn}{Definition}[section]

\newtheorem{rem}{Remark}[section]

\newcommand{\wtilde}[1]{\widetilde{#1}}
\newcommand{\et}{\textit{et~al.}}

\usepackage[pdftex]{graphicx}
\graphicspath{{./figs/}}
\usepackage[caption=false]{subfig}

\begin{document}
\title{Nonlinear Fourier Transform of Time-Limited and One-sided Signals}

\author{Vishal Vaibhav}
\email[]{vishal.vaibhav@gmail.com}
\noaffiliation{}
\date{\today}

\begin{abstract}
In this article, we study the properties of the nonlinear Fourier spectrum in
order to gain better control of the temporal support of the
signals synthesized using the inverse nonlinear Fourier transform (NFT). In
particular, we provide necessary and sufficient conditions satisfied by the
nonlinear Fourier spectrum such that the generated signal has a prescribed 
support. In our exposition, we assume that the support is a simply connected
domain that is either a bounded interval or the half-line, which amounts to studying the
class of signals which are either time-limited or one-sided, respectively. Further, it is shown that 
the analyticity properties of the scattering coefficients of the aforementioned
classes of signals can be exploited to improve the numerical conditioning of the differential 
approach of inverse scattering. Here, we also revisit the integral approach of
inverse scattering and provide the correct derivation
of the so called T\"oplitz inner-bordering algorithm. Finally, we conduct
extensive numerical tests in order to verify the analytical results presented in
the article. These tests also provide us an opportunity to compare the
performance of the two aforementioned numerical approaches in terms of accuracy
and complexity of computations.
\end{abstract}


\maketitle

\section*{Notations}
\label{sec:notations}
The set of non-zero positive real numbers ($\field{R}$) is denoted by
$\field{R}_+$. Real and imaginary parts of complex numbers ($\field{C}$) are
denoted by $\Re(\cdot)$ and $\Im(\cdot)$, respectively. The complex conjugate of 
$\zeta\in\field{C}$ is denoted by $\zeta^*$. The upper half (lower half) of $\field{C}$ 
is denoted by $\field{C}_+$ ($\field{C}_-$) and it closure by $\ovl{\field{C}}_+$
($\ovl{\field{C}}_-$). The support of a function
$f:\Omega\rightarrow\field{R}$ in $\Omega$ is defined as $\supp
f=\ovl{\{x\in\Omega|\,f(x)\neq0\}}$. The Lebesgue spaces of complex-valued 
functions defined in $\field{R}$ are denoted by 
$\fs{L}^p$ for $1\leq p\leq\infty$ with their corresponding 
norm denoted by $\|\cdot\|_{\fs{L}^p}$ or $\|\cdot\|_p$. The inverse 
Fourier-Laplace transform of a function $F(\zeta)$ analytic in 
$\ovl{\field{C}}_+$ is defined as
\[
{f}(\tau)=\frac{1}{2\pi}\int_{\Gamma}F(\zeta)e^{-i\zeta\tau}\,d\zeta,
\]
where $\Gamma$ is any contour parallel to the real line. The class of $m$-times 
differentiable complex-valued functions is denoted 
by $\fs{C}^m$. A function of class $\fs{C}^m$ is said to belong to 
$\fs{C}_0^m(\Omega)$, if the function and 
its derivatives up to order $m$ have a compact
support in $\Omega$ and if they vanish on the boundary ($\partial\Omega$).

\section{Introduction}
A nonlinear generalization of the conventional Fourier transform can be achieved 
via the two-component non-Hermitian Zakharov-Shabat (ZS) scattering 
problem~\cite{ZS1972}. As presented originally, it provides a means of solving 
certain nonlinear initial-value problems whose general description is provided 
by the AKNS-formalism~\cite{AKNS1974,AS1981}. The role of the nonlinear Fourier 
transform (NFT) here is entirely analogous to that of the Fourier transform in 
solving linear initial-value problems (IVPs) of dispersive nature. One of the
practical applications of this theory is in optical fiber communication where 
the master equation for propagation of optical field in a loss-less single 
mode fiber under Kerr-type focusing nonlinearity is the nonlinear 
Schr\"odinger equation (NSE)~\cite{HK1987,Agrawal2013}:
\begin{equation}\label{eq:NSE}
    i\partial_xq=\partial_t^2q+2|q|^2q,\quad(t,x)\in\field{R}\times\field{R}_+,
\end{equation}
where $q(t,x)$ is a complex valued function associated with the slowly varying
envelope of the electric field, $t$ is the retarded time and $x$ 
is the position along the fiber. 
There is a growing interest in the research community to exploit the nonlinear 
Fourier spectrum of the optical pulses for encoding information in an attempt to
mitigate nonlinear signal distortions at higher powers. 
This forms part of the motivation for this work where we study various signal
processing aspects of NFT that may be useful in any NFT-based modulation scheme. We 
refer the reader to a comprehensive review 
article~\cite{TPLWFK2017} and the references therein 
for an overview of NFT-based optical communication methodologies.
 
The ZS problem also appears naturally in certain physical systems. For instance, the
coupled-mode equations for co-propagating modes in a grating-assisted 
co-directional coupler (GACC), a device used to couple 
light between two different guided modes of an optical fiber, is a ZS problem
with the coupling coefficient as the potential (see~\cite{FZ2000,BS2003} and 
references therein). Note that the coupling
coefficient of physical GACCs must be compactly supported. In~\cite{BS2003},
the consequence of this requirement was studied in a discrete framework.

In this work, we formulate the necessary and sufficient conditions 
for the nonlinear Fourier spectrum (continuous as well as discrete) such that the 
corresponding signal has a prescribed support. In addition, certain useful regularity
properties of the continuous spectrum are proven which resemble that of the
conventional Fourier transform. For the conventional Fourier transform, these
results are already established in the well-known Paley-Wiener theorems. Note
that a straightforward way of generating time-limited signals is by windowing using a
rectangle function. In~\cite{V2018CNSNS}, exact solution of the scattering problem for 
a doubly-truncated multisoliton potential was presented. This task was
accomplished by first solving for the one-sided potential obtained as a result
of truncation from one side using the method discussed in the work of
Lamb~\cite{Lamb1980} (see also~\cite{RM1992,RS1994,SK2008}). Both of these
exactly solvable cases prove to be an insightful example of time-limited and
one-sided signals where analyticity properties of the scattering coefficients 
can be determined precisely.

The next goal that we pursue in this work is to exploit these regularity properties
to improve the numerical conditioning of the inverse scattering algorithms. In
particular, we revisit the \emph{differential approach}~\cite{V2017INFT1,V2017BL} 
as well as the \emph{integral approach}~\cite{BFPS2007,FBPS2015} of inverse scattering, 
and, discuss how to synthesize the input to these
algorithms such that the generated signal has the prescribed support. Further, we 
address in this work, the kind of numerical ill-conditioning that results from a pole
of the reflection coefficient that is too `close' to the real axis. It turns out
that the analyticity property of the reflection coefficient, $\rho(\zeta)$, in
$\field{C}_+$ can be exploited to develop a robust inverse scattering algorithm
within the differential approach that is based on the exponential \emph{trapezoidal
rule} presented in~\cite{V2017INFT1}. The advantage of the differential approach
is also seen in other numerical examples in this article where the
aforementioned second order convergent algorithm outperforms that based on the 
integral approach in terms of accuracy as well as complexity of computations.

The article is organized as follows: The main results of this paper for the
continuous-time NFT is presented in Sec.~\ref{sec:pptyNFS} where we focus on
certain class of time-limited and one-sided signals. In Sec.~\ref{sec:num-methods}, 
we discuss the two approaches for inverse scattering, namely, the differential 
and the integral approach. Sec.~\ref{sec:results} discusses all the numerical
results and Sec.~\ref{sec:conclusion} concludes this article.

\section{Continuous-Time Nonlinear Fourier Transform}
\label{sec:pptyNFS}
The starting point of our discussion is the ZS scattering problem for 
the complex-valued potentials $q(t,x)$ and $r(t,x)=-q^*(t,x)$ which 
in the AKNS-formalism (see~\cite{AKNS1974} for a complete introduction) can be stated as follows: 
Let $\zeta\in\field{R}$ and $\vv{v}=(v_1,v_2)^T$; then, for $t\in\field{R}$,  
\begin{equation}\label{eq:chi}
\partial_t{v}_1 = -i\zeta{v}_1+q{v}_2,
\quad\partial_t{v}_2 = i\zeta{v}_2+r{v}_1.
\end{equation}
In this article, we are mostly interested in studying the NFT for a fixed $x$;
therefore, we suppress any dependence on $x$ for the sake of brevity. The 
nonlinear Fourier spectrum can be defined using the scattering
coefficients, namely, $a(\zeta)$ and $b(\zeta)$, which are determined from the 
asymptotic form of the Jost solution,
$\vv{v}=\vs{\phi}(t;\zeta)$: As $t\rightarrow-\infty$, we have 
$\vs{\phi}e^{i\zeta t}\rightarrow(1,0)^{\tp}$, and, as $t\rightarrow\infty$, we have 
$\vs{\phi}\rightarrow(a(\zeta)e^{-i\zeta t}, b(\zeta)e^{i\zeta t})^{\tp}$.
An equivalent description of the system is obtained via the Jost solution 
$\vv{v}=\vs{\psi}(t;\zeta)$ whose asymptotic behavior as $t\rightarrow\infty$ is
$\vs{\psi}(t;\zeta)e^{-i\zeta t}\rightarrow(0,1)^{\tp}$, and, an equivalent set of
scattering coefficients can be determined from its asymptotic form (as
$t\rightarrow-\infty$) $\vs{\psi}\rightarrow(\ovl{b}(\zeta)e^{-i\zeta t}, a(\zeta)e^{i\zeta t})^{\tp}$
with the symmetry property $\ovl{b}(\zeta)=b^*(\zeta)$ ($\zeta\in\field{R}$).

In general, the nonlinear Fourier spectrum for the potential $q(t)$ comprises 
a \emph{discrete} and a \emph{continuous spectrum}. The discrete spectrum consists 
of the so called \emph{eigenvalues} $\zeta_k\in\field{C}_+$, such that 
$a(\zeta_k)=0$, and, the \emph{norming constants} $b_k$ such that 
$\vs{\phi}(t;\zeta_k)=b_k\vs{\psi}(t;\zeta_k)$. Note that $(\zeta_k,\,b_k)$
describes a \emph{bound state} or a \emph{solitonic state}
associated with the potential. For convenience, let the
discrete spectrum be denoted by the set
\begin{equation}
\mathfrak{S}_K=\{(\zeta_k,\,b_k)\in\field{C}^2|\,\Im{\zeta_k}>0,\,k=1,2,\ldots,K\}. 
\end{equation}
The continuous spectrum, also referred to as the \emph{reflection coefficient}, is 
defined by $\rho(\xi)={b(\xi)}/{a(\xi)}$ for $\xi\in\field{R}$.

Introducing the ``local'' scattering coefficients $a(t;\zeta)$ and 
$b(t;\zeta)$ such that 
$\vs{\phi}(t;\zeta)=(a(t;\zeta)e^{-i\zeta t}, b(t;\zeta)e^{i\zeta t})^{\tp}$, 
the scattering problem in~\eqref{eq:chi} reads as 
\begin{equation}\label{eq:ODE-AB-coeffs}
\begin{split}
&\partial_{t}a(t;\zeta)=q(t)b(t;\zeta)e^{2i\zeta t},\\
&\partial_{t}b(t;\zeta)=r(t)a(t;\zeta)e^{-2i\zeta t}.
\end{split}
\end{equation}
\subsection{The direct transform: Time-limited signals}\label{sec:tl-sig}
Let the scattering potential $q(t)$ be a time-limited signal with its support in 
$\Omega=[-T_-,T_+]$ where $T_{\pm}\geq0$. The 
initial conditions for the Jost solution $\vs{\phi}$
are: $a(-T_-;\zeta) = 1$ and  $b(-T_-;\zeta) = 0$. The scattering 
coefficients can be directly obtained from these functions as
$a(\zeta)=a(T_+;\zeta)$ and $b(\zeta)=b(T_+;\zeta)$. For the type of potential
at hand, these coefficients are known
to be analytic functions of $\zeta\in\field{C}$~\cite{AKNS1974}.

In this section, it will be useful to transform the ZS problem
in~\eqref{eq:ODE-AB-coeffs} as follows: Let us define, for the sake of
convenience, the modified Jost solution
\begin{equation}\label{eq:def-P}
\wtilde{\vv{P}}(t;\zeta) 
=\vs{\phi}(t;\zeta)e^{i\zeta t}-
\begin{pmatrix}
1\\
0
\end{pmatrix}
=\begin{pmatrix}
a(t;\zeta)-1\\
b(t;\zeta)e^{2i\zeta t}
\end{pmatrix},
\end{equation}
so that
$\wtilde{\vv{P}}(T_+;\zeta)e^{-2i\zeta T_+}
=([a(\zeta)-1]e^{-2i\zeta T_+},b(\zeta))^{\tp}$. The 
system of equations in~\eqref{eq:ODE-AB-coeffs} can be transformed into a set of 
Volterra integral equations of the second kind for $\wtilde{\vv{P}}(t;\zeta)$:
\begin{equation}\label{eq:volterra}
\wtilde{\vv{P}}(t;\zeta)=\vs{\Phi}(t;\zeta)
    +\int_{\Omega}\mathcal{K}(t,y;\zeta)\wtilde{\vv{P}}(y;\zeta)dy,
\end{equation}
where $\vs{\Phi}(t;\zeta)=(\Phi_1,\Phi_2)^{\tp}\in\field{C}^2$ with 
\begin{equation}
\begin{split}
\Phi_1(t;\zeta)&=\int_{-T_-}^{t}q(z)\Phi_2(z;\zeta)dz,\\
\Phi_2(t;\zeta)&=\int_{-T_-}^{t}r(y)e^{2i\zeta(t-y)}dy,
\end{split}
\end{equation}
and the Volterra kernel 
$\mathcal{K}(x,y;\zeta)=\diag(\mathcal{K}_1,\mathcal{K}_2)\in\field{C}^{2\times2}$
is such that
\begin{equation}\label{eq:volterra-kernel}
\begin{split}
&\mathcal{K}_1(x,y;\zeta) = r(y)\int_{y}^{x}q(z)e^{2i\zeta(z-y)}dz,\\
&\mathcal{K}_2(x,y;\zeta) = q(y)\int_{y}^{x}r(z)e^{2i\zeta(x-z)}dz,
\end{split}
\end{equation}
with $\mathcal{K}(x,y;\zeta)=0$ for $y>x$. In the following, we establish that 
$a(\zeta)$ and $b(\zeta)$ are of exponential type in
$\field{C}$. 


\begin{theorem}\label{thm:ab-estimate0}
Let $q\in\fs{L}^{1}$ with support in $\Omega$ and set
$\kappa=\|q\|_{\fs{L}^1(\Omega)}$. Then the estimates
\begin{align}
&|b(\zeta)|\leq
\sinh(\kappa)\times
\begin{cases}
    e^{2T_+\Im{\zeta}},&\zeta\in\ovl{\field{C}}_+,\\
    e^{-2T_-\Im{\zeta}},&\zeta\in{\field{C}}_-,
\end{cases}\label{eq:b-estimate-result}\\
&|\tilde{a}(\zeta)|\leq
[\cosh(\kappa)-1]\times
\begin{cases}
    e^{2T_+\Im{\zeta}},&\zeta\in\ovl{\field{C}}_+,\\
    e^{-2T_-\Im{\zeta}},&\zeta\in{\field{C}}_-,
\end{cases}\label{eq:a-estimate-result}
\end{align}
where $\tilde{a}(\zeta)$ denotes $[a(\zeta)-1]e^{-2i\zeta T_+}$, hold.
And, for fixed $\eta\in\field{R}$ such that $|\eta|<\infty$, we have
\begin{equation}\label{thm-ab-limit}
\lim_{\xi\in\field{R},\,|\xi|\rightarrow\infty}|f(\xi+i\eta)|=0,
\end{equation}
where $f(\zeta)$ denotes either $b(\zeta)$ or $\tilde{a}(\zeta)$.
\end{theorem}

\begin{proof}
The proof can be obtained using
the same method as in~\cite{AKNS1974}. For fixed
$\zeta\in\ovl{\field{C}}_+$, let $\OP{K}$ denote the Volterra integral operator
in~\eqref{eq:volterra} corresponding to the kernel $\mathcal{K}(x,y;\zeta)$ such
that
\begin{multline}
\OP{K}[\wtilde{\vv{P}}](t;\zeta)=\int_{\Omega}\mathcal{K}(t,y;\zeta)\wtilde{\vv{P}}(y;\zeta)dy\\
=\int_{-T_-}^tdz\int_{-T_-}^zdy
\begin{pmatrix}
q(z)r(y)e^{2i\zeta(z-y)}\wtilde{{P}}_1(y;\zeta)\\
q(y)r(z)e^{2i\zeta(t-z)}\wtilde{{P}}_2(y;\zeta)
\end{pmatrix}.
\end{multline}
The $\fs{L}^{\infty}(\Omega)$-norm~\cite[Chap.~9]{GLS1990} of
$\OP{K}_j,\,j=1,2,$ is given by
\begin{equation}
\|\OP{K}_j\|_{\fs{L}^{\infty}(\Omega)}=\esssup_{t\in\Omega}\int_{\Omega}|\mathcal{K}_j(t,y;\zeta)|dy,
\end{equation}
so that $\|\OP{K}_j\|_{\fs{L}^{\infty}(\Omega)}\leq\kappa^2/2$~\cite{AKNS1974}. 
The resolvent $\OP{R}_j$ of this operator exists and is given by the Neumann series
$\OP{R}_j=\sum_{n=1}^{\infty}\OP{K}^{(j)}_n$ where
$\OP{K}^{(j)}_n=\OP{K}_j\circ\OP{K}^{(j)}_{n-1}$ with
$\OP{K}^{(j)}_{1}=\OP{K}_j$. It can also be
shown using the methods in~\cite{AKNS1974} that
$\|\OP{K}^{(j)}_n\|_{\fs{L}^{\infty}(\Omega)}\leq{\kappa^{2n}}/{(2n)!}$, 
yielding the estimate 
$\|\OP{R}_j\|_{\fs{L}^{\infty}(\Omega)}\leq [\cosh(\kappa)-1]$.

For $j=2$, the solution of the Volterra integral equation can be stated as
\begin{equation}
\wtilde{P}_2(t;\zeta)={\Phi}_2(t;\zeta)+\OP{R}_2[{\Phi}_2](t;\zeta).
\end{equation}
From the estimate $\|\Phi_2(t;\zeta)\|_{\fs{L}^{\infty}}\leq\kappa$ and 
\begin{align*}
&\|\OP{K}^{(2)}_1[\Phi_2](t;\zeta)\|_{\fs{L}^{\infty}(\Omega)}\leq\frac{\kappa^3}{2\cdot3},\\
&\|\OP{K}^{(2)}_2[\Phi_2](t;\zeta)\|_{\fs{L}^{\infty}(\Omega)}\leq\frac{\kappa^5}{2\cdot3\cdot4\cdot5},\\
&\ldots,
\end{align*}
we have $\|\wtilde{P}_2(t;\zeta)\|_{\fs{L}^{\infty}}\leq\sinh\kappa$ for
$\zeta\in\ovl{\field{C}}_+$. For $j=1$, the solution of the Volterra integral equation can be stated as
\begin{equation}
\wtilde{P}_1(t;\zeta)={\Phi}_1(t;\zeta)+\OP{R}_1[{\Phi}_1](t;\zeta).
\end{equation}
From the estimate $\|\Phi_1(t;\zeta)\|_{\fs{L}^{\infty}}\leq\kappa^2/2$ and 
\begin{align*}
&\|\OP{K}^{(1)}_1[\Phi_1](t;\zeta)\|_{\fs{L}^{\infty}(\Omega)}
\leq\frac{\kappa^4}{2\cdot3\cdot4},\\
&\|\OP{K}^{(1)}_2[\Phi_1](t;\zeta)\|_{\fs{L}^{\infty}(\Omega)}
\leq\frac{\kappa^6}{2\cdot3\cdot4\cdot5\cdot6},\\
&\ldots,
\end{align*}
we have $\|\wtilde{P}_1(t;\zeta)\|_{\fs{L}^{\infty}}\leq[\cosh(\kappa)-1]$ for
$\zeta\in\ovl{\field{C}}_+$.

For the case $\zeta\in\field{C}_-$, we consider
$\wtilde{\vv{P}}_-(t;\zeta)=\wtilde{\vv{P}}(t;\zeta)e^{-2i\zeta t}$ so that
$\wtilde{\vv{P}}_-(T_+;\zeta)
=([a(\zeta)-1]e^{-2i\zeta T_+},b(\zeta))^{\tp}$. The 
Volterra integral equations for 
$\wtilde{\vv{P}}_-(t;\zeta)=(\wtilde{P}^{(-)}_1,\wtilde{P}^{(-)}_2)^{\tp}$ reads as:
\begin{equation}\label{eq:volterra-lh}
\wtilde{P}^{(-)}_j(t;\zeta)={\Phi}^{(-)}_j(t;\zeta)
    +\int_{\Omega}\mathcal{K}^{(-)}_j(t,y;\zeta)\wtilde{P}^{(-)}_j(y;\zeta)dy,
\end{equation}
where ${\Phi}^{(-)}_j(t;\zeta)={\Phi}_j(t;\zeta)e^{-2i\zeta t}$ 
and the Volterra kernels are given by
\begin{equation}\label{eq:volterra-kernel-lh}
\begin{split}
&\mathcal{K}^{(-)}_1(x,y;\zeta) = r(y)\int_{y}^{x}q(z)e^{-2i\zeta(x-z)}dz,\\
&\mathcal{K}^{(-)}_2(x,y;\zeta) = q(y)\int_{y}^{x}r(z)e^{-2i\zeta(z-y)}dz,
\end{split}
\end{equation}
with $\mathcal{K}_-(x,y;\zeta)=0$ for $y>x$. Using the approach outlined above,
it can be shown that, for $\zeta\in\field{C}_-$, 
$\|\wtilde{P}^{(-)}_1(t;\zeta)\|_{\fs{L}^{\infty}(\Omega)}\leq [\cosh(\kappa)-1]e^{-2\Im(\zeta)T_-}$ 
and $\|\wtilde{P}^{(-)}_2(t;\zeta)\|_{\fs{L}^{\infty}(\Omega)}\leq
\sinh(\kappa)e^{-2\Im(\zeta)T_-}$. 

Now combining the two cases, we obtain the
results~\eqref{eq:b-estimate-result} and~\eqref{eq:a-estimate-result}.

The result~\eqref{thm-ab-limit} follows from Lemma~\ref{lemma:RL-compact}.
The proof of this lemma is provided in the Appendix~\ref{app:RL-compact} 
which is similar to that of the Riemann-Lebesgue lemma~\cite[Chap.~13]{Jones2001}.
\end{proof}
\begin{lemma}\label{lemma:RL-compact} Let $q\in\fs{L}^{1}$ be supported in $\Omega$, then
\begin{equation}
\lim_{\xi\in\field{R},\,|\xi|\rightarrow\infty}\|{\Phi}_2(t;\xi+i\eta)\|_{\fs{L}^{\infty}(\Omega)}=0,
\end{equation}
for fixed $\eta\in\field{R}$ and $|\eta|<\infty$.
\end{lemma}

An immediate consequence of the preceding theorem together with the Paley-Wiener 
theorem~\cite[Chap.~VI]{Yosida1968} is that 
\begin{equation}
\supp\fourier^{-1}[f]\subset[-2T_+,2T_-].
\end{equation}
The Fourier transformation in the above equation is understood in the sense of distributions. 

Let us assume that $a(\zeta)$ has no zeros in
$\ovl{\field{C}}_+$. Consider the limit
$\lim_{r\rightarrow\infty}\log\left[1/|a(re^{i\theta})|\right]=0$ for
$0\leq\theta\leq\pi$. Then there exists a constant $M>0$ such that 
$|a(\zeta)|^{-1}\leq M,$ for $\zeta\in\ovl{\field{C}}_+$.
Therefore, $1/a(\zeta)$ is of exponential type zero in 
$\ovl{\field{C}}_+$\footnote{This can also be 
confirmed by the exact representation of $\log|a(\zeta)|$ which 
is guaranteed by~\cite[Thm~6.5.4]{Boas1954} in the upper-half plane:
\begin{equation}
\log|a(\zeta)|=\frac{\eta}{\pi}\int_{-\infty}^{\infty}
\frac{\log|a(\sigma)|}{(\sigma-\xi)^2+\eta^2}d\sigma,
\end{equation}
where $\zeta=\xi+i\eta$ with $\xi\in\field{R}$ and $\eta>0$.}. Then 
from Theorem~\ref{thm:ab-estimate0}, it follows that $\rho(\zeta)$ is of exponential
type $2T_+$ in $\ovl{\field{C}}_+$. In particular,  
\begin{equation}
|\rho(\zeta)|\leq CMe^{2T_+\Im{\zeta}},\quad\zeta\in\ovl{\field{C}}_+.
\end{equation}
Therefore, $\supp \fourier^{-1}[\rho]\subset[-2T_+,\infty)$. It is clear from the analysis 
above that the decay behavior of $\rho(\zeta)$ is 
largely dictated by the behavior of the scattering coefficient $b(\zeta)$.

Next, consider the situation when the discrete spectrum of a time-limited signal is
not empty. Let us mention that the nature of the discrete spectrum for time-limited 
signals is discussed in~\cite{AKNS1974}. To summarize these properties, let us observe 
that all the scattering coefficient, namely, $a(\zeta)$, $\ovl{a}(\zeta)$, $b(\zeta)$ 
and $\ovl{b}(\zeta)$ are entire functions of $\zeta$ with $\ovl{a}(\zeta) =
a^*(\zeta^*)$ and $\ovl{b}(\zeta)=b^*(\zeta^*)$. It is also straightforward to
conclude that the relationship 
\begin{equation}
a(\zeta)a^*(\zeta^*)+b(\zeta)b^*(\zeta^*)=1,
\end{equation}
holds for all $\zeta\in\field{C}$. For any eigenvalue $\zeta_k$, this relationship 
takes the form $b(\zeta_k)b^*(\zeta_k^*)=1$. Therefore, the norming constants
are given by $b_k=b(\zeta_k)=1/b^*(\zeta_k^*)$.
\begin{prop}
Let $(\zeta_k, b_k)$ be the eigenvalue and the corresponding norming constant
belonging to the discrete spectrum of a time-limited signal; then
\begin{equation}
b_k=b(\zeta_k)=1/b^*(\zeta_k^*).
\end{equation}
\end{prop}

In the following, the space of complex-valued functions of bounded 
variation over $\field{R}$ is denoted by $\fs{BV}$ and the variation of any
function $f\in\fs{BV}$ over $\Omega\subset\field{R}$ is denoted by
$\OP{V}[f;\Omega]$. If $q\in\fs{BV}$, then $\partial_tq\in\fs{L}^1$ 
exists almost everywhere such that 
$\|\partial_tq\|_{\fs{L}^1}\leq\OP{V}[q;\Omega]$~\cite[Chap.~16]{Jones2001}. Let 
$q^{(1)}$ be equivalent to $\partial_tq$ so that 
$\|q^{(1)}\|_{\fs{L}^1}=\|\partial_tq\|_{\fs{L}^1}$. 
\begin{prop}\label{prop:b-estimate}
Let $q\in\fs{BV}$ with support in $\Omega=[-T_-,T_+]$ such that it vanishes on $\partial\Omega$.
Define $\kappa=\|q\|_{\fs{L}^{1}}$ and 
$D=\frac{1}{2}\|q\|_{\fs{L}^{\infty}}+\|q\|_{\fs{L}^{1}}+\frac{1}{2}\|q^{(1)}\|_{\fs{L}^{1}}$; 
then, the estimate
\begin{equation}\label{prop:b-estimate-result}
|b(\zeta)|\leq
\frac{D\cosh(\kappa)}{1+|\zeta|}\times
\begin{cases}
    e^{2T_+\Im{\zeta}},&\zeta\in\ovl{\field{C}}_+,\\
    e^{-2T_-\Im{\zeta}},&\zeta\in{\field{C}}_-,
\end{cases}
\end{equation}
holds. Further, if $\beta(\tau)=\fourier^{-1}[b](\tau)$, then 
\begin{equation}
\supp\beta\subset[-2T_+, 2T_-],
\end{equation}
and $\beta\in\fs{L}^1\cap\fs{L}^2$.
\end{prop}
\begin{proof}
The proof follows the same line of reasoning as in
Theorem~\ref{thm:ab-estimate0}. Addressing the case $\zeta\in\ovl{\field{C}}_+$,
consider the second component in the Volterra integral
equation~\eqref{eq:volterra}:
\begin{equation}\label{eq:P2-resolve}
\wtilde{P}_2(t;\zeta)={\Phi}_2(t;\zeta)+\OP{K}_2[{\Phi}_2](t;\zeta).
\end{equation}
Recall that $\wtilde{P}_2(t;\zeta)=b(t;\zeta)e^{2i\zeta t}$ from~\eqref{eq:def-P}.
In the proof of Theorem~\ref{thm:ab-estimate0}, it was shown that the solution
of~\eqref{eq:P2-resolve} can be stated as
\begin{equation}
\wtilde{P}_2(t;\zeta)={\Phi}_2(t;\zeta)+\OP{R}_2[{\Phi}_2](t;\zeta),
\end{equation}
where $\OP{R}_2$ is the resolvent of the kernel $\mathcal{K}_2$ defined
in~\eqref{eq:volterra-kernel}. From the
observation $\|\OP{R}_2\|_{\fs{L}^{\infty}(\Omega)}\leq [\cosh(\kappa)-1]$, we
obtain the estimate
\begin{equation}\label{eq:estimate-P2}
\|\wtilde{P}_2(t;\zeta)\|_{\fs{L}^{\infty}(\Omega)}\leq
\cosh(\kappa)\|\Phi_2(t;\zeta)\|_{\fs{L}^{\infty}(\Omega)}.
\end{equation}
Now, for fixed $\zeta\in\ovl{\field{C}}_+$, using the fact that $q\in\fs{BV}(\Omega)$ and 
that it vanishes on $\partial\Omega$, it is straightforward 
to show, using integration by parts, that
\begin{equation}
\|{\Phi}_2(t;\zeta)\|_{\fs{L}^{\infty}}\leq\frac{D}{(1+|\zeta|)}.
\end{equation}
Plugging this result in~\eqref{eq:estimate-P2} yields the case 
$\zeta\in\ovl{\field{C}}_+$ in~\eqref{prop:b-estimate-result}.

Again, using integration by parts, for fixed $\zeta\in\field{C}_-$, it is also straightforward 
to show that 
\begin{equation}
\|{\Phi}^{(-)}_2(t;\zeta)\|_{\fs{L}^{\infty}(\Omega)}\leq
\frac{D}{(1+|\zeta|)}e^{-2T_-\Im\zeta}.
\end{equation}
From~\eqref{eq:volterra-lh}, an estimate of the form
\begin{equation}
\|b(t;\zeta)\|_{\fs{L}^{\infty}(\Omega)}\leq
\cosh(\kappa)\|\Phi^{(-)}_2(t;\zeta)\|_{\fs{L}^{\infty}(\Omega)},
\end{equation}
can be obtained which then yields the case $\zeta\in\field{C}_-$ 
in~\eqref{prop:b-estimate-result}. 

The estimate in~\eqref{prop:b-estimate-result} shows that $b(\zeta)$ is of
exponential-type in $\field{C}$ and the Paley-Wiener theorem~\cite[Chap.~VI]{Yosida1968}
ensures that $\supp\beta\subset[-2T_+, 2T_-]$. The last part of the 
result requires the part (a) of Lemma~\ref{lemma:inverse-FL}.
\end{proof}
For the sake of convenience, let us introduce the following class of functions:
\begin{defn}
A function $F(\zeta)$ is said to belong to the class $\fs{H}_+(T)$ if it is
analytic in $\ovl{\field{C}}_+$ and satisfies the following estimate
\[
|F(\zeta)|\leq\frac{C}{1+|\zeta|}e^{2T\Im\zeta},\quad\zeta\in\ovl{\field{C}}_+,
\]
for some constant $C>0$. Similarly, a function $F(\zeta)$ is said to belong to the class 
$\fs{H}_-(T)$ if it is
analytic in $\ovl{\field{C}}_-$ and satisfies the following estimate
\[
|F(\zeta)|\leq\frac{C}{1+|\zeta|}e^{-2T\Im\zeta},\quad\zeta\in\ovl{\field{C}}_-,
\]
for some constant $C>0$.
\end{defn}
Clearly, if $F(\zeta)\in\fs{H}_+(T)$, then $F^*(\zeta^*)\in\fs{H}_-(T)$. The inverse 
Fourier-Laplace Transform of such functions is studied in the
following lemma which is proved in the Appendix~\ref{app:inverse-FL}:
\begin{lemma}\label{lemma:inverse-FL}
Let $F\in\fs{H}_+(T)$ and define $f(\tau)=\fourier^{-1}[F](\tau)$. 
\begin{itemize}
\item[$(a)$]The function $f(\tau)$ is supported in 
$[-2T,\infty)$ and belongs to $\fs{L}^1\cap\fs{L}^2$.

\item[$(b)$] Define $G(\zeta)=i\zeta F(\zeta)-\mu e^{-2i\zeta T}$ where $\mu$ is such
that $G(\zeta)e^{2i\zeta T}\rightarrow0$ as $|\zeta|\rightarrow\infty$ in
$\ovl{\field{C}}_+$ and $G\in\fs{H}_+(T)$. Let $g(\tau)=\fourier^{-1}[G](\tau)$; then, 
$g(\tau)$ is supported in $[-2T,\infty)$ and belongs to $\fs{L}^1\cap\fs{L}^2$.
Furthermore, $f(\tau)=\mu+\int^{\tau}_{-2T}g(\tau')d\tau'$ 
so that it has a finite limit as $\tau\rightarrow-2T+0$ and it is bounded on
$[-2T,\infty)$.
\end{itemize}
\end{lemma}

Strengthening the regularity conditions on the potential allows us to strengthen
the regularity results for $\beta(\tau)$ as evident from the following theorem:
\begin{theorem}\label{thm:b-nth-derivative}
Let $m$ be a finite positive integer. Let
$q\in\fs{C}^{m+1}_0(\Omega)$. Then the function $\beta(\tau)=\fourier^{-1}[b](\tau)$ is such that
$\supp\beta(\tau)\subset[-2T_+,2T_-]$ and it is $m$-times differentiable in 
$(-2T_+,2T_-)$. Furthermore, $\partial^k_{\tau}\beta(\tau)\rightarrow0$ as 
$\tau\rightarrow 2T_--$ and $\tau\rightarrow-2T_++$ for all $0\leq k\leq m$.
\end{theorem}
\begin{proof}
The proof can be obtained by a repeated application of the
Prop.~\ref{prop:b-estimate} as follows: We consider the 
second component of the Volterra integral equations~\eqref{eq:volterra} and~\eqref{eq:volterra-lh}. 
Let $R_0(t)=r(t)$ and define a sequence of functions
\begin{equation}
R_{k+1}(t) = \partial_tR_{k}(t)-r(t)\int_{-T_-}^tq(y)R_{k}(y)dy.
\end{equation}
It can be shown that $R_k(t)$ is bounded on $\Omega$ for $k\leq m$. 
Further, set $\wtilde{P}_{2,0}(t;\zeta)=\wtilde{P}_{2}(t;\zeta)={b}(t;\zeta)e^{2i\zeta t}$
and introduce a sequence of functions $\wtilde{P}_{2,k}(t;\zeta)$ as
\begin{equation}
\wtilde{P}_{2,k+1}(t;\zeta)=i\zeta \wtilde{P}_{2,k}(t;\zeta)+\frac{1}{2}R_{k}(t)
\end{equation}
where $k\leq m$. Using mathematical induction, it can be shown 
that each of the iterates satisfy the integral equations given by
\begin{equation}
\wtilde{P}_{2,k}(t;\zeta)
={\Phi}_{2,k}(t;\zeta)+\int_{\Omega}\mathcal{K}_2(t,y;\zeta)\wtilde{P}_{2,k}(y;\zeta)dy,
\end{equation}
where $0\leq k\leq m+1$ with 
\begin{equation}
\Phi_{2,k}(t;\zeta)=\frac{1}{2}\int_{-T_-}^{t}R_k(y)e^{2i\zeta (t-y)}dy.
\end{equation}
Following the approach used in Prop.~\ref{prop:b-estimate}, it can be shown
that $b_k(\zeta)=\wtilde{P}_{2,k}(T_+;\zeta)e^{-2i\zeta T_+}$ belongs to
$\fs{H}_+(T_+)\cup\fs{H}_-(T_-)$.
Define $\beta_k(\tau)=\fourier^{-1}[{b}_{k}](\tau),\,k>0$, 
then it is evident that its support falls in $[-2T_+,2T_-]$ (Paley-Wiener theorem).

Also, note that each of the pairs ${b}_{k}(\zeta)$ and
${b}_{k+1}(\zeta)$
for $0\leq k\leq m$ satisfy the conditions of the Lemma~\ref{lemma:inverse-FL}
(with $F=b_k(\zeta)$ and $G=b_{k+1}(\zeta)$); therefore,
$\partial_{\tau}\beta_k(\tau)=\beta_{k+1}(\tau)$ almost everywhere.
This implies that $\beta(\tau)$ is $m$-times differentiable with 
$\partial^k_{\tau}\beta(\tau)=\beta_k(\tau)$ for $0\leq k\leq m$.

The limit $\lim_{\tau\rightarrow-2T_++0}\beta_k(\tau)=-({1}/{2})R_{k}(T_+)=0$
follows from the fact that
\begin{equation}
\lim_{\zeta\in\field{C}_+,\,\zeta\rightarrow\infty}
i\zeta b_k(\zeta)=-\frac{1}{2}R_k(T_+).
\end{equation}
Now consider $\ovl{b}_k=b^*_k(\zeta^*)$ so that
$\fourier^{-1}[\ovl{b}](\tau)=\beta^*(-\tau)$. The limit
$\lim_{\tau\rightarrow-2T_-+0}\beta^*_k(-\tau)=({1}/{2})R^*_{k}(T_+)=0$
follows from the fact that
\begin{equation}
\lim_{\zeta\in\field{C}_+,\,\zeta\rightarrow\infty}
i\zeta\ovl{b}_k(\zeta)=\frac{1}{2}R^*_k(T_+).
\end{equation}
Therefore, $\beta_k(2T_--0)=0$.
\end{proof}
\begin{rem}\label{rem:rho-nth-derivative}
Putting $\wtilde{\Omega}=[-2T_+, 2T_-]$, the 
preceding theorem states that if $q\in\fs{C}^{m+1}_0(\Omega)$, then
$\beta\in\fs{C}^{m}_0(\wtilde{\Omega})$. It also follows that for every
non-negative integer $k\leq m$ there exists a $D_k>0$ such that
\begin{equation}\label{eq:estimate-b-pw}
    |b(\zeta)|\leq\frac{D_k}{(1+|\zeta|)^k}\times
\begin{cases}
    e^{2T_+\Im{\zeta}},&\zeta\in\ovl{\field{C}}_+,\\
    e^{-2T_-\Im{\zeta}},&\zeta\in{\field{C}}_-.
\end{cases}
\end{equation}
Again, assuming that $a(\zeta)$ has no zeros in
$\ovl{\field{C}}_+$, there exists $M>0$ such that
\begin{equation}\label{eq:rho-estimate-decay}
    |\rho(\zeta)|\leq\frac{D_kM}{(1+|\zeta|)^{k}}e^{2T_+\Im{\zeta}},
\quad\zeta\in\ovl{\field{C}}_+.
\end{equation}
\end{rem}

\subsubsection{Example: Doubly-truncated one-soliton
potential}\label{sec:one-soliton}
This example is taken from~\cite{V2018CNSNS}. Let $q(t)$ denote the one-soliton 
potential with the discrete spectrum $(\zeta_1,b_1)$ where $\zeta=\xi_1+i\eta_1$ so that 
\begin{equation}
q(t)=\frac{4\eta_1\beta_{0}}{1+|\beta_{0}|^2},
\end{equation}
where $\beta_{0}(t;\zeta_1, b_1) = - (1/b_{1})e^{-2i\zeta_1t}$. We make the potential 
compactly supported by truncating the part which lies outside 
$[-T_-, T_+]$. Let the doubly-truncated
version of $q(t)$ be denoted by $q^{(\sqcap)}(t)$. Note that, unlike $q(t)$, the
doubly-truncated potential $q^{(\sqcap)}(t)$ is not reflectionless.

Let $2T=T_++T_-$ and define
$Z_+=1/\beta_0(T_+)$ and $Z_-=\beta_0(-T_-)$ so that
$|Z_{\pm}|=|b_1|^{\pm1}e^{-2\eta_1 T_{\pm}}$.
The scattering coefficients can be
shown to be given by
\begin{equation}
\begin{split}
a^{(\sqcap)}(\zeta)&=1+
\frac{2i\eta_1Z^*_{+}Z^*_{-}}{\Xi}
\left[\frac{e^{4i\zeta T}-e^{4i\zeta^*_1 T}}{\zeta-\zeta^*_1}
-\frac{e^{4i\zeta T}-e^{4i\zeta_1 T}}{\zeta-\zeta_1}\right],
\end{split}
\end{equation}
and 
\begin{equation}
\begin{split}
b^{(\sqcap)}(\zeta) &= \frac{2i\eta_1b_1|Z_-|^2}{\Xi}
\frac{\left(e^{-2i(\zeta-\zeta_1)T_+}-e^{2i(\zeta-\zeta_1)T_-}\right)}{\zeta-\zeta_1}\\
&\quad+\frac{2i\eta_1|Z_+|^2}{b^*_1\Xi}
\frac{\left(e^{-2i(\zeta-\zeta^*_1)T_+}-e^{2i(\zeta-\zeta^*_1)T_-}\right)}{\zeta-\zeta^*_1},
\end{split}
\end{equation}
where $\Xi=(1+|Z_+|^2)(1+|Z_{-}|^2)$. Putting $2T_0 = T_+-T_-$, we have
\begin{equation}
\begin{split}
\beta^{(\sqcap)}(\tau)= \frac{2\eta_1}{\Xi}e^{-i\xi_1\tau}
\left[b_1|Z_-|^2e^{\eta_1\tau}
+\frac{|Z_+|^2}{b^*_1}e^{-\eta_1\tau}\right]\Pi\left(\frac{\tau+2T_0}{2T}\right),
\end{split}
\end{equation}
where $\Pi(\tau)$ denotes the rectangle function defined as
\begin{equation}
\Pi(\tau)=
\begin{cases}
1,&|\tau|\leq 1,\\
0,&\text{otherwise}.
\end{cases}
\end{equation}

\subsection{The direct transform: One-sided signals}
Let $\Omega=(-\infty,T_+]$ in the following paragraphs unless stated otherwise.
For $d>0$, consider the class of complex-valued
functions supported in $\Omega$, denoted by $\fs{E}_d(\Omega)$, such that for
$f\in\fs{E}_d(\Omega)$, there exists $C>0$ such that the estimate 
$|{f}(t)|{\leq}Ce^{-2d|t|}$ holds almost everywhere in $\Omega$.
Clearly, $\fs{E}_d(\Omega)\subset\fs{L}^{p}(\Omega)$ for
$1\leq p\leq\infty$.

For $q\in\fs{E}_d(\Omega)$, the scattering coefficients, $a(\zeta)$ and $b(\zeta)$, are
analytic in $\ovl{\field{C}}_+$, (in fact the region of analyticity turns out to be
$\Im{\zeta}>-d$)~\cite{AKNS1974}. The results in Theorem~\ref{thm:ab-estimate0} 
can now be modified as follows:
\begin{theorem}\label{thm:ab-estimate0-exp}
Let $q\in\fs{E}_d(\Omega)$ with support in $\Omega$ and set
$\kappa=\|q\|_{\fs{L}^1(\Omega)}$. Then the estimates
\begin{align}
&|b(\zeta)|\leq\sinh(\kappa)e^{2T_+\Im{\zeta}},\label{eq:b-estimate-exp-result}\\
&|\tilde{a}(\zeta)|\leq[\cosh(\kappa)-1]e^{2T_+\Im{\zeta}},\label{eq:a-estimate-exp-result}
\end{align}
where $\tilde{a}(\zeta)$ denotes $[a(\zeta)-1]e^{-2i\zeta T_+}$, hold for $\zeta\in\ovl{\field{C}}_+$.
And, for fixed $0\leq\eta<\infty$, we have
\begin{equation}\label{thm-ab-limit-exp}
\lim_{\xi\in\field{R},\,|\xi|\rightarrow\infty}|f(\xi+i\eta)|=0,
\end{equation}
where $f(\zeta)$ denotes either $b(\zeta)$ or $\tilde{a}(\zeta)$.
\end{theorem}
An immediate consequence of the preceding theorem is that 
\begin{equation}
\supp\fourier^{-1}[f]\subset[-2T_+,\infty).
\end{equation}
The Fourier transformation in the above equation is understood in the sense of distributions. 

When $a(\zeta)$ has no zeros in $\ovl{\field{C}}_+$, it follows from the 
Theorem~\ref{thm:ab-estimate0-exp} that $\rho(\zeta)$ is of exponential
type $2T_+$ in $\ovl{\field{C}}_+$, i.e., there exists a constant $C>0$ such
that
\begin{equation}
|\rho(\zeta)|\leq Ce^{2T_+\Im{\zeta}},\quad\zeta\in\ovl{\field{C}}_+.
\end{equation}
Therefore, $\supp\fourier^{-1}[\rho]\subset[-2T_+,\infty)$ which is the same
result one would obtain if the signal was time-limited. Further, the fact that 
$b(\zeta)$ is analytic in $\ovl{\field{C}}_+$ implies that, for any eigenvalue $\zeta_k$, 
the corresponding norming constant is given by $b_k=b(\zeta_k)$~\cite{AKNS1974}.

We conclude this section with a modification of the Prop.~\ref{prop:b-estimate}
for the class of one-sided signals introduced above:
\begin{prop}\label{prop:b-estimate-exp}
Let $q\in\fs{E}_d(\Omega)\cap\fs{BV}(\Omega)$, and, define 
$\kappa=\|q\|_{\fs{L}^{1}(\Omega)}$ and $D=\frac{1}{2}\|q\|_{\fs{L}^{\infty}(\Omega)}+
\|q\|_{\fs{L}^{1}(\Omega)}+\frac{1}{2}\|q^{(1)}\|_{\fs{L}^{1}(\Omega)}$. Then the estimate
\begin{equation}\label{prop:b-estimate-exp-result}
|b(\zeta)|\leq\frac{D\cosh(\kappa)}{1+|\zeta|}e^{2T_+\Im{\zeta}},\quad\zeta\in\ovl{\field{C}}_+,
\end{equation}
holds. Further, if $\beta(\tau)=\fourier^{-1}[b](\tau)$, then 
\begin{equation}
\supp\beta\subset[-2T_+,\infty),
\end{equation}
and $\beta\in\fs{L}^1\cap\fs{L}^2$. 
\end{prop}

\subsubsection{Example: Truncated one-soliton potential}
This example has been previously treated
in~\cite{Lamb1980,RM1992,RS1994,SK2008,V2018CNSNS}. Consider the one-soliton potential 
discussed in Sec.~\ref{sec:one-soliton}. We make the potential one-sided by truncating 
the part which lies outside 
$(-\infty, T_+]$. Let the one-sided version of $q(t)$ be denoted by 
$q^{(-)}(t)$. Note that, unlike $q(t)$, the one-sided potential $q^{(-)}(t)$ 
is not reflectionless. Using the quantities defined in
Sec.~\ref{sec:one-soliton}, the scattering coefficients work out to be
\begin{equation}
a^{(-)}(\zeta) =\frac{|Z_+|^2}{1+|Z_+|^2}
+\frac{1}{1+|Z_+|^2}\left(\frac{\zeta-\zeta_1}{\zeta-\zeta_1^*}\right),
\end{equation}
and
\begin{equation}
b^{(-)}(\zeta)=\frac{Z_+}{1+|Z_+|^2}
\left(\frac{2i\eta_1}{\zeta-\zeta_1^*}\right)e^{-2i\zeta T_+}.
\end{equation}
Let $\theta(\tau)$ be the Heaviside step-function, then
\begin{equation}
\beta^{(-)}(\tau)=\frac{2\eta_1Z_+}{1+|Z_+|^2}e^{-i\zeta_1^*(\tau+2T_+)}\theta(\tau+2T_+).
\end{equation}

\subsection{The inverse problem}
In this section, we would like to address the converse of some of the results obtained 
above. Define the nonlinear impulse response as 
\begin{equation}
p(\tau) = \fourier^{-1}[\rho](\tau)
=\frac{1}{2\pi}\int_{-\infty}^{\infty}\rho(\xi) e^{-i\xi\tau}d\xi.
\end{equation}
In the following, we would like to examine the solution of the inverse scattering problem when 
$p(\tau)\in\fs{L}^{1}\cap\fs{L}^{2}$ with support in $[-2T_+,\infty)$. We first
give an important result due to Epstein which estimates the energy in the tail
of the scattering potential defined by
\begin{equation}\label{eq:energy-tail-invscatter}
\mathcal{E}_+(T)=\int^{\infty}_{T}|q(t)|^2dt.
\end{equation}
For a given $p(\tau)\in\fs{L}^{1}\cap\fs{L}^{2}$, define 
\begin{equation}
\mathcal{I}_m(T)=\left[\int^{\infty}_{2T}|p(-\tau)|^md\tau\right]^{1/m};
\end{equation}
then, we have the following result:
\begin{prop}[Epstein~\cite{Epstein2004}]\label{prop:epstien0}
For a given $p(\tau)\in\fs{L}^{1}\cap\fs{L}^{2}$, if there exists a time 
$T$ such that $\mathcal{I}_1(T)<1$; then, the estimate
\begin{equation*}
\mathcal{E}_+(T)\leq\frac{2\mathcal{I}^2_2(T)}{[1-\mathcal{I}^2_1(T)]},
\end{equation*}
holds provided $p(\tau)$ is continuous in the half-space $(-\infty,-2T)$.
\end{prop}
\begin{proof}
Consider the
Jost solutions with prescribed asymptotic behavior as 
$x\rightarrow\infty$:
\begin{equation}
\vs{\psi}(t;\zeta)=
\begin{pmatrix}
0\\
1
\end{pmatrix}e^{i\zeta t}
+\int_t^{\infty}e^{i\zeta s}{\vv{A}}(t,s)ds,
\end{equation}
where $\vv{A}$ is independent of $\zeta$. Consider the Gelfand-Levitan-Marchenko (GLM) integral
equations. In the following we fix
$t\in\field{R}$ so that the GLM equations for $y\in\Omega_t=[t,\infty)$ is given by 
\begin{equation}\label{eq:GLM-start}
\begin{split}
&{A}_2^*(t,y)=-\int_{t}^{\infty}{A}_1(t,s){f}(s+y)ds,\\
&{A}_1^*(t,y) = f(t+y) +\int_{t}^{\infty}A_2(t,s){f}(s+y)ds,
\end{split}
\end{equation}
where $f(\tau)=p(-\tau)$. The solution of the GLM equations allows us to recover
the scattering potential using $q(t)=-2A_1(t,t)$ together with the estimate 
$\|q\chi_{[t,\infty)}\|^2_2=-2A_2(t,t)$ where $\chi_{\Omega}$ denotes the
characteristic function of $\Omega\subset\field{R}$. Define the operator 
$\OP{K}$ as 
\begin{equation}
\begin{split}
\OP{K}[g](y)
&=\int_t^{\infty}ds\int_t^{\infty}dx\,f^*(y+s)f(s+x)g(x)\\
&=\int_t^{\infty}\mathcal{K}(y,x;t)g(x)dx,
\end{split}
\end{equation}
where the kernel function $\mathcal{K}(y,x;t)$ is given by
\begin{equation}
\mathcal{K}(y,x;t)
=\int_t^{\infty}ds\,f^*(y+s)f(s+x).
\end{equation}
The GLM equations in~\eqref{eq:GLM-start} can now be stated as
\begin{equation}\label{eq:fredholm}
{A}_j(t,y)={\Phi}_j(t,y)-\OP{K}[{A}_j(t,\cdot)](y),\quad j=1,2,
\end{equation}
which is a Fredholm integral equation of the second kind where
\begin{equation}
\begin{split}
{\Phi}_1(t,y)&=f^*(t+y),\\
{\Phi}_2(t,y)&=-\int_t^{\infty}f^*(y+s)f(t+s)ds.
\end{split}
\end{equation}
Recalling $\mathcal{I}_m(t)=\|f\chi_{[2t,\infty)}\|_{\fs{L}^m}$ for
$m=1,2,\infty$, we have
\begin{equation}
\begin{split}
\|\OP{K}\|_{\fs{L}^{\infty}(\Omega_t)}
&=\esssup_{y\in\Omega_t}\int_{t}^{\infty}dx\,|\mathcal{K}(y,x;t)|\\
&\leq\esssup_{y\in\Omega_t}\int_{t}^{\infty}dx\,\int_t^{\infty}ds\,|f(y+s)||f(s+x)|\\
&\leq\esssup_{y\in\Omega_t}\int_{t+y}^{\infty}du|f(u)|\,\int_{t+u-y}^{\infty}du_1\,|f(u_1)|\\
&\leq [\mathcal{I}_1(t)]^2,
\end{split}
\end{equation}
and, $\|\Phi_2(t,\cdot)\|_{\fs{L}^{\infty}(\Omega_t)}\leq[\mathcal{I}_{2}(t)]^2$.
If $\mathcal{I}_1(t)<1$, then the standard theory of Fredholm equations suggests that the
resolvent of the operator $\OP{K}$ exists~\cite{GLS1990}. Under this assumption, 
from~\eqref{eq:fredholm}, we have
\begin{multline*}
\|{A}_j(t,\cdot)\|_{\fs{L}^{\infty}(\Omega_t)}\leq
\|\Phi_j(t,\cdot)\|_{\fs{L}^{\infty}(\Omega_t)}\\
+\|\OP{K}\|_{\fs{L}^{\infty}(\Omega_t)}\|{A}_j(t,\cdot)\|_{\fs{L}^{\infty}(\Omega_t)},
\end{multline*}
which yields
\begin{equation}
\begin{split}
&\|{A}_1(t,\cdot)\|_{\fs{L}^{\infty}(\Omega_t)}\leq\frac{\mathcal{I}_{\infty}(t)}{[1-\mathcal{I}^2_1(t)]},\\
&\|{A}_2(t,\cdot)\|_{\fs{L}^{\infty}(\Omega_t)}\leq\frac{\mathcal{I}^2_2(t)}{[1-\mathcal{I}^2_1(t)]}.
\end{split}
\end{equation}
Given that from here one can only assert that 
$|A_j(t,y)|\leq\|{A}_j(t,\cdot)\|_{\fs{L}^{\infty}(\Omega_t)}$ almost everywhere
(a.e.), we need to ascertain the continuity of $A_j(t,y)$ with respect to $y$ 
throughout the domain $\Omega_t$
or as $y\rightarrow t$ from above. Assume that $f(\tau)$ is continuous, then
$\Phi_j(t,y)$ is continuous with respect to $y$. It can be seen that the kernel function
$\mathcal{K}(y,x;t)$ is also continuous with respect to $y$. Therefore, if the resolvent 
kernel is continuous (w.r.t. $y$) then the result 
follows. To this end, consider the Neumann series for the resolvent
$\OP{R}=\sum_{n\in\field{Z}_+}(-1)^n\OP{K}_n$ where
$\OP{K}_n=\OP{K}\circ\OP{K}_{n-1}$ with $\OP{K}_1=\OP{K}$. For fixed
$t$, the partial sums 
$\sum_{1\leq n\leq N}\|\OP{K}_n\|_{\fs{L}^{\infty}(\Omega_t)}
\leq[1-\mathcal{I}^2_1(t)]^{-1}$ for all $N<\infty$.
Therefore, uniform convergence of the partial sums allows us to conclude the continuity of the 
limit of the partial sums.
\end{proof}

\begin{corr}
Consider $p(\tau)\in\fs{L}^{1}\cap\fs{L}^{2}$ with support in
$[-2T_+,\infty)$. If the solution of the GLM equation~\eqref{eq:GLM-start}
exists, then
\begin{equation}
\supp q\subset (-\infty,T_+].
\end{equation}
\end{corr}
\begin{proof}
The proof is an immediate consequence of Prop.~\ref{prop:epstien0} and the fact that 
$\mathcal{I}_j(t)=0,\,j=1,2$ for all $t\in[T_+,\infty)$.
\end{proof}
\begin{rem}
The requirement that $p(\tau)\in\fs{L}^{1}\cap\fs{L}^{2}$ in the corollary above can
be weakened as long as the existence of the GLM equations can be guaranteed. Setting 
$f(\tau)=p(-\tau)$ and observing that $\supp f\subset (-\infty, 2T_+]$, we may write 
the GLM equations as  
\begin{equation*}
\begin{split}
&{A}_2^*(t,y)=-\int_{t}^{2T_+-y}{A}_1(t,s){f}(s+y)ds,\\
&{A}_1^*(t,y) = f(t+y) +\int_{t}^{2T_+-y}A_2(t,s){f}(s+y)ds.
\end{split}
\end{equation*}
If the solution $A_j(t,y)$ exists, then $\supp_y A_j(t,y)\subset [t,2T_+-t]$, $j=1,2$.
Therefore, $A_j(t,y)\equiv0$ for $t>T_+$ yielding $\supp q\subset (-\infty,T_+]$.
\end{rem}
Existence of the solution of the GLM equation~\eqref{eq:GLM-start} for a given 
$p(\tau)\in\fs{L}^{1}\cap\fs{L}^{2}$ when $\mathcal{I}_1(T)>1$ is established in~\cite{Epstein2004}. The
proof proceeds by observing that $\OP{K}$ is a self-adjoint compact operator
such that $\|\OP{K}\|_{\fs{L}^{2}(\Omega_t)}\leq [\mathcal{I}_1(t)]^2$ which
follows from
\begin{equation*}
\begin{split}
&\|\OP{K}[g]\|^2_{\fs{L}^{2}(\Omega_t)}\\
&=\int_t^{\infty}dy\left|\int_{t}^{\infty}dx\,\mathcal{K}(y,x;t)g(x)\right|^2\\
&\leq\int_{t}^{\infty}dy\left[\int_{t}^{\infty}dx\,|\mathcal{K}(y,x;t)|
\int_{t}^{\infty}dx\,|\mathcal{K}(y,x;t)||g(x)|^2\right]\\
&\leq [\mathcal{I}_1(t)]^4\|g\|^2_{\fs{L}^{2}(\Omega_t)},
\end{split}
\end{equation*}
For 
\[
\left(1+[I_1(t)]^2\right)^{-1}<\lambda_t<{2}\left(1+[I_1(t)]^2\right)^{-1},
\] 
define 
\begin{equation}
\OP{T}[g] = \lambda_t \OP{K}[g]-(1-\lambda_t)g,
\end{equation}
so that the Fredholm equation in~\eqref{eq:fredholm} can be written as
\begin{equation}\label{eq:fredholm2}
{A}_j(t,y)=\lambda_t{\Phi}_j(t,y)-\OP{T}[{A}_j(t,\cdot)](y),\quad j=1,2.
\end{equation}
The resolvent of the operator $\OP{T}$ exists as a consequence of the fact that
$\|\OP{T}\|_2<1$. If $f(\tau)$ is continuous, the continuity of ${A}_j(t,y)$ for
fixed $t$ follows from the uniform convergence of the Neumann series. 

We state the following result which specifies the sufficient conditions for the
$\beta(\tau)$ for the scattering potential to be compactly supported:
\begin{theorem}\label{thm:beta-support-q}
Let $\beta(\tau)=\fourier^{-1}[b](\tau)\in\fs{BV}$ with its support in
$\wtilde{\Omega}=[-2T_+,2T_-]$ such that it vanishes on 
$\partial\wtilde{\Omega}$. If $|b(\xi)|<1$ for $\xi\in\field{R}$, then, 
there exists a unique soliton-free scattering potential $q\in\fs{L}^{\infty}$ such that 
\[
\supp q\subset[-T_-,T_+].
\] 
\end{theorem}
\begin{proof}
The condition in the first part of the theorem guarantees that
$p(\tau)\in\fs{L}^1\cap\fs{L}^2$ is continuous with $\supp p\subset[-2T_+,\infty)$. The proof
for $\supp q\subset (-\infty,T_+]$ follows from the preceding discussion. This
also shows that $\supp q^*(-t)\subset (-\infty,T_-]$ by considering $b^*(\xi)$ 
in place of $b(\xi)$. Combining the two cases allows us to conclude that $\supp
q\subset[-T_-,T_+]$.

Now, it remains to show that $q\in\fs{L}^{\infty}$. Setting $f(\tau)=p(-\tau)$ and considering the support
of $f(\tau)$, we may write the GLM equations as  
\begin{equation}
\begin{split}
&{A}_2^*(t,y)=-\int_{t}^{2T_+-y}{A}_1(t,s){f}(s+y)ds,\\
&{A}_1^*(t,y) = f(t+y) +\int_{t}^{2T_+-y}A_2(t,s){f}(s+y)ds.
\end{split}
\end{equation}
This allows us to conclude that $\supp_y A_1(t,y)\subset [t,2T_+-t]$. 
Furher, from the continuity of $A_1(t,y)$ wih respect to $y$ over the bounded set $[t,2T_+-t]$, it follows
that $q(t)=-2A_1(t,t)$ is bounded.
\end{proof}

Note that the discussion in the preceding paragraphs where the starting point of the
discussion was $p(\tau)$ does not require explicit assumption about the
presence/absence of the bound states. However, when starting with $\beta(\tau)$, this
assumption is required in order to construct $p(\tau)$.
\subsubsection{Presence of bound states}
Let us start our discussion with one-sided signals supported in $\Omega=(-\infty,T_+]$. In 
this case, $b(\zeta)$ is analytic in $\ovl{\field{C}}_+$ and, for any eigenvalue $\zeta_k$, 
the corresponding norming constant is given by $b_k=b(\zeta_k)$~\cite{AKNS1974}.
Let us show that this condition is sufficient to guarantee the one-sided support
of $q(t)$. To this end, we first prove the following lemma which specifies the
support of $p(\tau)$: 
\begin{lemma}\label{lemma:one-sided-S-R}
Let $b(\zeta)\in\fs{H}_+(T_+)$ with $|b(\xi)|<1$ for
$\xi\in\field{R}$. Let the discrete spectrum be given by 
$\mathfrak{S}_K=\{(\zeta_k,b_k)|,\,b_k=b(\zeta_k),\,k=1,2,\ldots,K\}$, then
\begin{equation}\label{eq:supp-p}
\supp p\subset [-2T_+,\infty).
\end{equation}
\end{lemma}
\begin{proof}
For $\kappa>0$, consider the contours
$\Gamma_{\kappa}=\{\zeta\in\field{R}|\,|\zeta|\leq \kappa\}$ and
$C_{\kappa}=\{\zeta\in\ovl{\field{C}}_+|\,|\zeta|=\kappa\}$. The radiative part
of the reflection coefficient is given by $\rho_R(\zeta)=\rho(\zeta)a_S(\zeta)$ where 
$a_S(\zeta)=\prod_{k=1}^K{(\zeta-\zeta_k)}/{(\zeta_k-\zeta_k^*)}$.
Then the condition on $b(\zeta)$ guarantees that $\rho_R\in\fs{H}_+(2T_+)$.
For sufficiently large $\kappa$, using Cauchy's theorem, we have
\begin{equation}
\mathcal{I}_{\kappa}(\tau)=\frac{1}{2\pi i}
\oint_{\Gamma_{\kappa}\cup
C_{\kappa}}\rho(\zeta)e^{i\zeta\tau}=\sum_k\Res[\rho;\zeta_k]e^{i\zeta_k\tau},
\end{equation}
where the contour $\Gamma_{\kappa}\cup C_{\kappa}$ is oriented
positively. For $\tau>2T_+$, the integrand on $C_{\kappa}$ is exponentially decaying with respect
to $\kappa$. Therefore, in the limit $\kappa\rightarrow\infty$, for $\tau>2T_+$, 
\begin{equation}
f(\tau)=\frac{1}{2\pi}
\int_{\field{R}}\rho(\zeta)e^{i\zeta\tau}-i\sum_k\Res[\rho;\zeta_k]e^{i\zeta_k\tau}=0,
\end{equation}
where we have used the fact that
$\Res[\rho;\zeta_k]=b(\zeta_k)/\dot{a}(\zeta_k)$. Now, the
conclusion~\eqref{eq:supp-p} follows from $p(\tau)=f(-\tau)$.
\end{proof}
This lemma ensures that $\supp q\subset (-\infty,2T_+]$. However, given the
exponentially increasing behavior of $p(\tau)$, one can only assert that
$\|q\chi_{[t,T_+]}\|_2$ is finite for finite $t$. This question will be addressed
later using a rather direct approach where such results can be obtained trivially. 
\begin{rem}
Define $\sigma(\zeta)=b^*(\zeta^*)/a(\zeta)$. Proceeding as in Lemma~\ref{lemma:one-sided-S-R}, 
it can also be concluded that if $b(\zeta)\in\fs{H}_-(T_-)$ and $|b(\xi)|<1$
then $1/b_k=b^*(\zeta_k^*)$ implies $s(\tau)=\fourier^{-1}[\sigma](\tau)$ is supported in 
$[-2T_-,\infty)$. 
\end{rem}
The remark shows that $\supp q^*(-t)\subset [-2T_-,\infty)$. Combining
the two cases, it immediately follows that $\supp q\subset[-T_-,T_+]$. 
Note that similar arguments can be found in~\cite[App.~5]{AKNS1974} where the
authors have characterized the nature of the nonlinear Fourier spectrum
corresponding to a time-limited signal. The
difference merely lies in what is considered as the starting point. Starting
from a one-sided $p(\tau)$ does not always lead to a compactly supported
$\beta(\tau)$; therefore, it is much more convenient to start with an
exponential type $b(\zeta)$ or compactly supported $\beta(\tau)$. However, 
this does not guarantee that $b(\zeta_k)=b_k$ and $b(\zeta^*_k)=1/b^*_k$ 
are satisfied at the same time. It appears that without losing a certain degree
of freedom in choosing $\beta(\tau)$, arbitrary bound states cannot be
introduced in the scattering potential while maintaining its compact support.

Now let us discuss a direct approach which also provides some insight into the 
effective support of potential obtained as a result of addition of bound states to 
compactly supported (radiative) potentials. The Darboux transformation (DT)
technique allows one to introduce bound states to any arbitrary potential referred
to as the \emph{seed} potential. In the
following discussion, we assume that we have a compactly supported potential, $q_0(t)$, with its
support in $[-T_-, T_+]$. For the sake of simplicity let us assume that the 
discrete spectrum of the seed potential is empty. The Jost solution of the seed
potential, in matrix form, is denoted by $v_0=(\vs{\phi}_0,\vs{\psi}_0)$. It is 
known that DT can be implemented as a recursive scheme~\cite{GHZ2005} which is briefly summarized 
below. Let us define the successive discrete spectra
$\emptyset=\mathfrak{S}_0\subset\mathfrak{S}_1\subset\mathfrak{S}_2
\subset\ldots\subset\mathfrak{S}_K$ such that 
${\mathfrak{S}}_j=\{(\zeta_j,b_j)\}\cup{\mathfrak{S}}_{j-1}$ for
$j=1,2,\ldots,K$ where $(\zeta_j,b_j)$ are distinct elements of 
$\mathfrak{S}_K$. The Darboux matrices of degree one can be
stated as
\begin{equation}
D_1(t;\zeta,\mathfrak{S}_{j}|\mathfrak{S}_{j-1})=\zeta\sigma_0 -
\begin{pmatrix}
\frac{|\gamma_{j-1}|^2\zeta_j+\zeta_j^*}{1+|\gamma_{j-1}|^2} 
&\frac{(\zeta_j-\zeta_j^*)\gamma_{j-1}}{1+|\gamma_{j-1}|^2}\\
\frac{(\zeta_j-\zeta_j^*)\gamma^*_{j-1}}{1+|\gamma_{j-1}|^2}
&\frac{\zeta_j+\zeta_j^*|\gamma_{j-1}|^2}{1+|\gamma_{j-1}|^2}
\end{pmatrix},
\end{equation}
where
\begin{equation}
\gamma_{j-1}(t;\zeta_j, b_j) =
\frac{\phi_1^{(j-1)}(t;\zeta_j)-b_{j}\psi_1^{(j-1)}(t;\zeta_j)}
{\phi_2^{(j-1)}(t;\zeta_j) - b_{j}\psi_2^{(j-1)}(t;\zeta_j)},
\end{equation}
for $(\zeta_j,b_j)\in\mathfrak{S}_K$ and the successive Jost solutions, 
${v}_{j} = (\vs{\phi}_{j},\vs{\psi}_{j})$, needed in this ratio are computed as
\begin{equation}
    {v}_j(t;\zeta)=\frac{1}{(\zeta-\zeta^*_j)}D_{1}(t;\zeta,\mathfrak{S}_j|\mathfrak{S}_{j-1})
    v_{j-1}(t;\zeta).
\end{equation}
Let us denote the successive augmented potentials by $q_j(t)$ and define 
\begin{equation}
\begin{split}
\mathcal{E}^{(-)}_j(t)&=\int^t_{-\infty}|q_j(s)|^2ds,\\
\mathcal{E}^{(+)}_j(t)&=\int_t^{\infty}|q_j(s)|^2ds.
\end{split}
\end{equation}
Then the potential is given by 
\begin{equation}\label{DT-iter-pot}
q_j = q_{j-1} -
2i\frac{(\zeta_j-\zeta_j^*)\gamma_{j-1}}{1+|\gamma_{j-1}|^2},
\end{equation}
and
\begin{equation}\label{eq:Epm-iter}
\begin{split}
\mathcal{E}^{(-)}_j &=\mathcal{E}^{(-)}_{j-1}
+\frac{4\Im(\zeta_j)}{1+|\gamma_{j-1}|^{-2}},\\
\mathcal{E}^{(+)}_j &=\mathcal{E}^{(+)}_{j-1}
+\frac{4\Im(\zeta_j)}{1+|\gamma_{j-1}|^{2}}.
\end{split}
\end{equation}
The above relations can be readily verified by computing the coefficient of
$\zeta^{0}$ and $\zeta^{-1}$ in the asymptotic expansion of
$v_j(t,\zeta)e^{i\sigma_3\zeta t}$ in the
negative powers of $\zeta$~\cite{AKNS1974}. Now, if we were to truncate the augmented potential 
such that the support becomes
$[-T'_-,T'_+]$ where $T_{\pm}\leq T'_{\pm}$, the energy in the tails can be
computed exactly thanks to the recurrence relations for $\mathcal{E}^{(\pm)}_j$.
Note that, for $T_-<t\leq-T'_-$, the seed Jost solution can be stated exactly: 
\begin{equation}
\vs{\phi}_0(t;\zeta) = 
\begin{pmatrix}
1\\
0
\end{pmatrix}e^{-i\zeta t},\quad
\vs{\psi}_0(t;\zeta)=\begin{pmatrix}
\ovl{b}_0(\zeta)e^{-i\zeta t}\\
a_0(\zeta)e^{i\zeta t}
\end{pmatrix},
\end{equation}
where $a_0(\zeta)$ and $b_0(\zeta)$ are the scattering coefficients of $q_0(t)$. Thus 
the energy content of the part of the signal supported in
$(-\infty,-T'_-]$, denoted by $\mathcal{E}^{(-)}(-T'_-)$, using the recursive 
scheme stated above in~\eqref{eq:Epm-iter}. Now let us examine what happens when
we assume $\ovl{b}(\zeta_k)=1/b_k$. For $k=1$, we have $\gamma_1=0$ so that
\[
{v}_1(t;\zeta)=
\begin{pmatrix}
1&0\\
0&\frac{\zeta-\zeta_1}{\zeta-\zeta_1^*}
\end{pmatrix}
v_{0}(t;\zeta),
\]
so that
\begin{equation}
\vs{\phi}_1(t;\zeta) = 
\begin{pmatrix}
1\\
0
\end{pmatrix}e^{-i\zeta t},\quad
\vs{\psi}_1(t;\zeta)=\begin{pmatrix}
\ovl{b}_0(\zeta)e^{-i\zeta t}\\
a_1(\zeta)e^{i\zeta t}
\end{pmatrix},
\end{equation}
where 
\begin{equation}
a_1(\zeta) = \left(\frac{\zeta-\zeta_1}{\zeta-\zeta_1^*}\right)a_0(\zeta).
\end{equation}
Using mathematical induction it can be shown that $\gamma_j=0$ for
$j=0,1,\ldots,K$, and
\begin{equation}
\vs{\phi}_j(t;\zeta) = 
\begin{pmatrix}
1\\
0
\end{pmatrix}e^{-i\zeta t},\quad
\vs{\psi}_j(t;\zeta)=\begin{pmatrix}
\ovl{b}_0(\zeta)e^{-i\zeta t}\\
a_j(\zeta)e^{i\zeta t}
\end{pmatrix},
\end{equation}
where 
\begin{equation}
a_j(\zeta) = \left(\frac{\zeta-\zeta_j}{\zeta-\zeta_j^*}\right)a_{j-1}(\zeta).
\end{equation}
Note that the arguments above are identical to those presented in~\cite{Lin1990}
where $T'_{\pm}=\infty$. From here it is easy to conclude that
$\mathcal{E}^{(-)}(-T'_-)\equiv0$. For $t=T'_+$, similar arguments can be
provided to conclude that, if $b(\zeta_k)=b_k$, the energy in the tail supported
in $[T'_+,\infty)$, denoted by $\mathcal{E}^{(+)}(T'_+)$, is identically zero.

Let us conclude this section with the following observation: Even if the relations 
$b(\zeta_k)=b_k$ and $\ovl{b}(\zeta_k)=1/b_k$
are not satisfied, the augmented signal, $q_K(t)$, can still be considered as effectively
supported within some $[-T'_-,T'_+]$ up to a tolerance, say $\epsilon$, such
that $\epsilon {\|q_K\|^2_2}=
{\mathcal{E}^{(-)}(-T'_-)+\mathcal{E}^{(+)}(T'_+)}$. Here, 
$\|q_K\|_2^2=\mathcal{E}^{(-)}(T'_+)+\mathcal{E}^{(+)}(T'_+)$ 
with $\mathcal{E}^{(-)}_0(T'_+)=\|q_0\|_2^2$ which is assumed to be known.
\section{Numerical Methods}
\label{sec:num-methods}
Numerical methods for inverse scattering can be grouped into two classes: The
first class of methods can be characterized as the \emph{differential approach} 
which proceeds by discretizing the ZS problem using exponential one-step methods
that admit of a layer-peeling property. The 
discrete framework thus obtained can be used for direct as well as 
inverse scattering\footnote{Note that not all exponential integrators admit of a
layer-peeling property and their characterization remains an open problem.
Besides, it's also not clear if the layer-peeling property is a necessary
condition for such systems to admit of an inverse scattering algorithm.}. The second class 
of methods can be characterized as 
the \emph{integral approach} which proceeds by dicretizing the 
Gelfand-Levitan-Marchenko equations. The discrete framework thus obtained can 
be used for direct as well as inverse scattering with some 
limitations\footnote{The presence of bound states
makes the nonlinear impulse response an increasing function of $\tau$ (the
covariable of $\xi$) in a certain half-space
making the GLM equations ill-conditioned for numerical computations. Therefore, the GLM based
approach is used in this article only when the discrete spectrum is empty.}.

\subsection{Differential approach}
\label{sec:num-disc}
In this section, we review the discretization scheme for the ZS problem first 
proposed in~\cite{V2017INFT1}. This scheme can be described as the exponential 
one-step method based on the trapezoidal rule of integration. In order to 
discuss the discretization scheme, we take an equispaced grid defined 
by $t_n= T_1 + nh,\,\,n=0,1,\ldots,N,$ with $t_{N}=T_2$ where $h$ is the grid spacing.
Define $\ell_-,\ell_+\in\field{R}$ such that $h\ell_-= T_-$, $h\ell_+= T_+$.
Further, let us define $z=e^{i\zeta h}$. For the potential functions sampled on the 
grid, we set $q_n=q(t_n)$, $r_{n}=r(t_n)$. Using the same convention, $U_{n}=U(t_n)$ and 
$\wtilde{U}_{n}=\wtilde{U}(t_n)$ where $\wtilde{U}$ is given by
\begin{equation}\label{eq:exp-int}
\wtilde{U}=e^{i\sigma_3\zeta t}Ue^{-i\sigma_3\zeta t}
=\begin{pmatrix}
0 & qe^{2i\zeta t}\\
re^{-2i\zeta t} & 0
\end{pmatrix}.
\end{equation}
Define $\tilde{\vv{\phi}}=(a(t;\zeta), b(t;\zeta))^{\tp}$. By applying the trapezoidal rule 
to~\eqref{eq:ODE-AB-coeffs}, we obtain
\begin{equation*}
\begin{split}
\tilde{\vs{\phi}}_{n+1} 
= \left(\sigma_0-\frac{h}{2}\wtilde{U}_{n+1}\right)^{-1}
\left(\sigma_0+\frac{h}{2}\wtilde{U}_{n}\right)\tilde{\vs{\phi}}_{n}.
\end{split}
\end{equation*}
Setting $2Q_{n}=hq_n$, $2R_{n}=hr_n$ and $\Theta_n=1-Q_nR_n$ so that
\begin{equation}\label{eq:scatter-TR}
\begin{split}
\vs{\phi}_{n+1}&=\frac{z^{-1}}{\Theta_{n+1}}
\begin{pmatrix}
1+z^2Q_{n+1}R_n& z^2Q_{n+1}+Q_n\\
R_{n+1}+z^2R_n & R_{n+1}Q_n + z^2
\end{pmatrix}\vv{v}_n\\
&=z^{-1}M_{n+1}(z^2)\vs{\phi}_n,
\end{split}
\end{equation}
or, equivalently,
\begin{equation}
\begin{split}
\vs{\phi}_n &=
\frac{z^{-1}}{\Theta_{n}}
\begin{pmatrix}
R_{n+1}Q_n + z^2& -z^2Q_{n+1}-Q_n\\
-R_{n+1}-z^2R_n & 1+z^2Q_{n+1}R_n
\end{pmatrix}\vs{\phi}_{n+1}\\
&=z^{-1}\wtilde{M}_{n+1}(z^2)\vs{\phi}_{n+1}.
\end{split}
\end{equation}
In order to express the discrete approximation to the Jost solutions, let us
define the vector-valued polynomial
\begin{equation}\label{eq:poly-vec}
\vv{P}_n(z)=\begin{pmatrix}
            P^{(n)}_{1}(z)\\
            P^{(n)}_{2}(z)
        \end{pmatrix}
         =\sum_{k=0}^n
            \vv{P}^{(n)}_{k}z^k
         =\sum_{k=0}^n
        \begin{pmatrix}
            P^{(n)}_{1,k}\\
            P^{(n)}_{2,k}
        \end{pmatrix}z^k.
\end{equation}
The Jost solutions $\vs{\phi}$ can be written in the form
$\vs{\phi}_n = z^{\ell_-}z^{-n}\vv{P}_n(z^2)$. Note that this expression 
corresponds to the boundary
condition $\vs{\phi}_0=z^{\ell_-}(1,0)^{\tp}$ which translate to 
$\vv{P}_0=(1,0)^{\tp}$. The other Jost
solution, $\ovl{\vs{\phi}}_n$, can be written as
$\ovl{\vs{\phi}}_n = z^{-\ell_-}z^{n}(i\sigma_2)\vv{P}^*_n(1/z^{*2})$.
The recurrence relation for the polynomial associated with the Jost solution
takes the form
\begin{equation}\label{eq:poly-scatter}
\vv{P}_{n+1}(z^2) = M_{n+1}(z^2)\vv{P}_n(z^2).
\end{equation}
The discrete approximation to the scattering coefficients is obtained from the scattered
field, $\vs{\phi}_{N}=(a_{N} z^{-\ell_+},b_{N} z^{\ell_+})^{\tp}$, which yields
\begin{equation}
a_{N}(z^2)={P}^{(N)}_1(z^2),\quad b_{N}(z^2)=(z^2)^{-\ell_{+}}{P}^{(N)}_2(z^2).
\end{equation}
The quantities $a_{N}$ and $b_{N}$ above are referred to as the 
\emph{discrete scattering coefficients}. Note that these coefficients can only be 
defined for $\Re\zeta\in [-{\pi}/{2h},\,{\pi}/{2h}]$.

\subsubsection{The layer-peeling algorithm}
\label{sec:discrete-TR-summary}
Let us consider the problem of recovering the samples of the
scattering potential from the discrete scattering coefficients known in the
polynomial form. This step is referred to as the \emph{discrete inverse
scattering} step. Starting from the recurrence relation~\eqref{eq:poly-scatter},
the \emph{layer-peeling} step consists in using 
$\vv{P}_{n+1}(z^2)$ to retrieve the samples of the potential
needed to compute the transfer matrix $\wtilde{M}_{n+1}(z^2)$ so that the entire step can be
repeated with $\vv{P}_{n}(z^2)$ until all the samples of the potential are 
recovered.
 
Let us assume $Q_0=0$. The recurrence relation for the trapezoidal rule yields
\begin{equation}
\label{eq:TR-cond}
P^{(n+1)}_{1,0}
=\Theta^{-1}_{n+1}\prod_{k=1}^{n}\biggl(\frac{1+Q_kR_k}{1-Q_kR_k}\biggl)
=\Theta^{-1}_{n+1}\prod_{k=1}^{n}\biggl(\frac{2-\Theta_k}{\Theta_k}\biggl)
\end{equation}
and $\vv{P}^{(n+1)}_{n+1}= 0$. The last relationship follows from the assumption $Q_0=0$. For sufficiently 
small $h$, it is reasonable to assume that 
$1+Q_nR_n=2-\Theta_n>0$ so that $P^{(n)}_{1,0}>0$ (it also implies that 
$|Q_n|=|R_n|<1$). Now for the layer-peeling step, we have the following relations
$R_{n+1} = {P^{(n+1)}_{2,0}}/{P^{(n+1)}_{1,0}}$ and 
\begin{equation}
R_n = \frac{\chi}{1 + \sqrt{1+|\chi|^2}},\quad
\chi=\frac{P^{(n+1)}_{2,1} - R_{n+1}P^{(n+1)}_{1,1}}
{P^{(n+1)}_{1,0} - Q_{n+1}P^{(n+1)}_{2,0}}.
\end{equation}
Note that $P^{(n+1)}_{1,0}\neq0$ and ${P^{(n+1)}_{1,0} -
Q_{n+1}P^{(n+1)}_{2,0}}\neq0$.
As evident from~\eqref{eq:scatter-TR}, the transfer matrix, $M_{n+1}(z^2)$, 
connecting $\vv{P}_n(z^2)$ and $\vv{P}_{n+1}$ is completely determined by
the samples $R_{n+1}$ and $R_n$ (with $Q_{n+1}=-R^*_{n+1}$ and $Q_{n}=-R^*_{n}$).

\subsubsection{Discrete inverse scattering with the reflection coefficient}
\label{sec:lp-rho}
Let us assume that the reflection coefficient $\rho\in\fs{H}_+(T_+)$ and 
define $\breve{\rho}(\zeta)=\rho(\zeta)e^{2i\zeta T_+}$. In order to obtain the discrete 
scattering coefficients, we follow a method due to
Lubich~\cite{Lubich1988I,Lubich1994} which 
is used in computing quadrature weights for convolution integrals. 
Introduce the function $\delta(z)$ as in~\cite{Lubich1988I,Lubich1994} which corresponds to
the trapezoidal rule of integration: $\delta(z) = 2{(1-z)}/{(1+z)}$. Next, we 
introduce the coefficients $\breve{\rho}_k$ such that
\begin{equation}
\breve{\rho}\biggl(\frac{i\delta(z^2)}{2h}\biggl)=\sum_{k=0}^{\infty}\breve{\rho}_kz^{2k},
\end{equation}
then using the Cauchy integral, we have the relation
\begin{equation}\label{eq:discrete-coeffs-rho}
\breve{\rho}_{k} = \frac{1}{2\pi i}
    \oint_{|z|^2=\varrho}\biggl[\breve{\rho}\biggl(\frac{i\delta(z^2)}{2h}\biggl)\biggl]z^{-2k-2}d(z^2),
\end{equation}
which can be easily computed using FFT. The input to the layer-peeling algorithm can 
then be taken to be
\begin{equation*}
P^{(N)}_1(z^2) = 1,\quad
P^{(N)}_2(z^2) =
\left\{\breve{\rho}\biggl(\frac{i\delta(z^2)}{2h}\biggl)\right\}_N,
\end{equation*}
where the notation $\{\cdot\}_{N}$ implies truncation after $N$ terms.

If $\rho(\zeta)$ decays sufficiently fast, then an alternative approach can be
devised. Define the function $\breve{p}(\tau)$ as 
\begin{equation}\label{eq:breve-p-tau}
\breve{p}(\tau) = \frac{1}{2\pi}\int_{-\infty}^{\infty}\breve{\rho}(\xi)
e^{-i\xi\tau}d\xi.
\end{equation}
Note that for $\tau < 0$, the contour can be closed in the upper half-plane and
the integrals would evaluate to zero; therefore, $\breve{p}(\tau)$ is causal. If, 
in addition, $\breve{\rho}(\zeta)$ is analytic in the strip
$-\kappa\geq\Im{\zeta}<0$, the nonlinear impulse response $\breve{p}(\tau)$ can
be shown to decay exponentially away from the origin. A polynomial approximation 
for $\breve{\rho}(\zeta)$ can be furnished by using the trapezoidal sum to 
approximate the integral
\begin{equation}
\breve{\rho}(\zeta) = \int_{0}^{\infty}\breve{p}(\tau) e^{i\zeta\tau}d\tau.
\end{equation}
Let the grid spacing in the
$\tau$-domain be $2h$, then
\begin{equation}
\begin{split}
\breve{\rho}(\zeta) 
&\approx h\breve{p}(0) + 2h\sum_{k=1}^{N-1}\breve{p}(2hk) [e^{2i\zeta h}]^k\\
&=h\breve{p}(0)+2h\sum_{k=1}^{N-1}\breve{p}(2hk)z^{2k}.
\end{split}
\end{equation}
Note that the endpoint $\tau_{N}=2Nh$ is omitted because $\breve{p}(\tau_N)$ is assumed to be
vanishingly small. For fixed $\zeta\in\ovl{\field{C}}_+$, the approximation is
$\bigO{h^2}$ assuming that the error due to truncation at $\tau_N$ is negligible
at least for $N>N_0$ where $N_0$ depends on the decay properties of
$\breve{p}(\tau)$ and $h$. The approximation above asserts that
$\breve{\rho}(\xi)$ is effectively supported in $[0,\pi/h]$. Note that the
method discussed in~\cite{V2017BL} makes the assumption that $\rho(\xi)$ has a
compact support which is independent of $h$.

If the integral for $\breve{p}(\tau)$ cannot be evaluated exactly, we use 
the FFT algorithm to compute the integral. The decay property of
$\breve{\rho}(\xi),\,\xi\in\field{R},$ decides 
the convergence of the discrete Fourier sum used in FFT so that it may 
be necessary to employ large number of samples in order to get
accurate results.

\subsection{Integral approach}
\label{sec:modified-TIB}
Let us assume that the reflection coefficient $\rho\in\fs{H}_+(T_+)$ and 
define $\breve{\rho}(\zeta)=\rho(\zeta)e^{2i\zeta T_+}$ as in
Sec.~\ref{sec:lp-rho}. Introducing
\begin{equation}
\begin{split}
&\mathcal{A}_1(t,y) = A_1(T_+-t,T_+-y+t),\\
&\mathcal{A}_2(t,y) = A_2(T_+-t,T_+-t+y),
\end{split}
\end{equation}
the GLM equations can be put into the following form
\begin{equation}\label{eq:GLM-mod}
\begin{split}
&\mathcal{A}_2^*(t,2t-y)+\int_{0}^{y}\mathcal{A}_1(t,2t-s)\breve{p}(y-s)ds=0,\\
&-\mathcal{A}_1^*(t,y)+\int_{0}^{y}\mathcal{A}_2(t,s)\breve{p}(y-s)ds=-\breve{p}(y),
\end{split}
\end{equation}
for $0\leq y\leq 2t\leq 2T=T_-+T_+$ where $\breve{p}(\tau)$ is defined
by~\eqref{eq:breve-p-tau}. The scattering potential is recovered using
\begin{equation}
\begin{split}
q(T_+-t) 
&= -2A_1(T_+-t,T_+-t)\\
&= -2\mathcal{A}_1(t,2t),
\end{split}
\end{equation}
with the $\fs{L}^2$-norm of the potential truncated to $(T_+-t,T_+]$ is given
by
\begin{equation}
\begin{split}
\int_0^{t}|q(T_+-s)|^2ds 
&= -2A_2(T_+-t,T_+-t)\\
&= -2\mathcal{A}_2(t,0).
\end{split}
\end{equation}
Defining the grid $y_{j+1}=y_j+2h$ and setting $\breve{p}_k=\breve{p}(2kh)$, we
have
\begin{align*}
\mathcal{A}_2^*(t,2t-y_l)
&+2h\sum_{j=0}^l\breve{p}_{l-j}\mathcal{A}_1(t,2t-y_{j})
-h\breve{p}_{0}\mathcal{A}_1(t,2t-y_l)\\
&-h\breve{p}_{l}\mathcal{A}_1(t,2t)=0,\\
\mathcal{A}_1^*(t,y_k)
&-2h\sum_{j=0}^k\breve{p}_{k-j}\mathcal{A}_2(t,y_{j})
+h\breve{p}_{0}\mathcal{A}_2(t,y_k)\\
&+h\breve{p}_{k}\mathcal{A}_2(t,0)
=\breve{p}_k.
\end{align*}
Setting $t_j = jh$ for $j=0,\ldots,N$ with 
$k,l\leq m$ and introducing the quadrature weights
$\omega_k=2h\breve{p}_k,\,k>0$ with $\omega_0=h\breve{p}_0$, we have
\begin{align*}
&\mathcal{A}_2(t_m,y_l)
+\sum_{j=l}^m{\omega}^{\dagger}_{l-j}\mathcal{A}^*_1(t_m,y_{j})
-\frac{1}{2}{\omega}^{\dagger}_{l}\mathcal{A}^*_1(t_m,y_m)=0,\\
-&\mathcal{A}_1^*(t_m,y_k)
+\sum_{j=0}^k{\omega}_{k-j}\mathcal{A}_2(t_m,y_{j})
-\frac{1}{2}{\omega}_{k}\mathcal{A}_2(t_m,0)
=-\breve{p}_k,
\end{align*}
where ${\omega}^{\dagger}_{l-j}={\omega}^{*}_{j-l}$. 
Now, define the 
lower triangular T\"oplitz matrix
$\Gamma^{(m)}=(\gamma^{(m)}_{jk})_{(m+1)\times(m+1)}$ where
\begin{equation*}
\gamma^{(m)}_{jk}=
\begin{cases}
{\omega}_{j-k},&j-k\geq0,\\
0,&\text{otherwise},
\end{cases}\quad j,k=0,1,\ldots,m,
\end{equation*}
and the T\"oplitz matrix $G^{(m)}$ as
\begin{equation*}
G^{(m)} = 
\begin{pmatrix}
    I^{(m)}&-\Gamma^{(m)\dagger}\\
    \Gamma^{(m)} & I^{(m)}
\end{pmatrix},
\end{equation*}
where $I^{(m)}$ is the $(m+1)\times (m+1)$ identity matrix and `$\dagger$'
denotes Hermitian conjugate. Next, we define the vectors 
\begin{equation*}
\vv{\mathcal{A}}^{(m)}_j
=(\mathcal{A}_j(t_m,0),\mathcal{A}_j(t_m,y_1),\ldots\mathcal{A}_j(t_m,y_m))^{\tp},\quad j=1,2,
\end{equation*}
and put
\begin{equation}
\vs{\chi}^{(m)} = 
\begin{pmatrix}
    \vv{\mathcal{A}}^{(m)}_2\\
    -\vv{\mathcal{A}}^{(m)*}_1
\end{pmatrix};
\end{equation}
then the linear system corresponding to the GLM equation reads as
\begin{equation}\label{eq:linear-toplitz}
G^{(m)}\vs{\chi}^{(m)} = 
\begin{pmatrix}
-h{\chi}^{(m)}_{2m+1}\breve{\vs{p}}^{(m)\ddagger}\\
-\breve{\vv{p}}^{(m)}+h{\chi}^{(m)}_{0}{\breve{\vs{p}}}^{(m)}
\end{pmatrix},
\end{equation}
where ${\vs{\omega}}^{(m)\ddagger}$ is the reverse numerated and conjugated form of
${\vs{\omega}}^{(m)}$. The double-dagger operation (${}^{\ddagger}$) will be
used to mean reversal of index numeration and complex conjugation in the rest of this article unless
otherwise stated. In order to develop a fast method of solution of the linear
system, Belai~\et~\cite{BFPS2007} exploited the T\"oplitz structure of the matrix
$G^{(m)}$. It is known that the inverse of a T\"oplitz matrix is a
\emph{persymmetric} matrix. Our aim is to determine 
$\chi^{(m+1)}_{0}$ and $\chi^{(m+1)}_{2m+3}$, the first and the last
component of the vector $\vs{\chi}^{(m+1)}$, respectively. Therefore, we focus
on the first and 
the last column of the matrix $[G^{(m)}]^{-1}$ which must be of the form
\begin{equation}
\begin{pmatrix}
\vv{u}^{(m)}\\
\vv{v}^{(m)}
\end{pmatrix},\quad
\begin{pmatrix}
-\vv{v}^{(m)\ddagger}\\
\vv{u}^{(m)\ddagger}
\end{pmatrix}.
\end{equation}
The complex conjugate of the first and the last row of $[G^{(m)}]^{-1}$ have the
form
\begin{equation}
\begin{pmatrix}
\vv{u}^{(m)}\\
-\vv{v}^{(m)}
\end{pmatrix}^{\tp},\quad
\begin{pmatrix}
\vv{v}^{(m)\ddagger}\\
\vv{u}^{(m)\ddagger}
\end{pmatrix}^{\tp}.
\end{equation}
Therefore, it follows that
\begin{equation}\label{eq:inverse-toeplitz}
G^{(m)}\begin{pmatrix}
\vv{u}^{(m)}\\
\vv{v}^{(m)}
\end{pmatrix}=\begin{pmatrix}
    1\\
    0\\
    \vdots\\
    0
\end{pmatrix},\quad
G^{(m)}\begin{pmatrix}
-\vv{v}^{(m)\ddagger}\\
\vv{u}^{(m)\ddagger}
\end{pmatrix}=\begin{pmatrix}
    0\\
    \vdots\\
    0\\
    1
\end{pmatrix}.
\end{equation}
These relations translate to
\begin{equation}\label{eq:inner-bordering-basic}
\begin{split}
&u^{(m)}_0-{\vs{\omega}}^{(m)*}\cdot\vv{v}^{(m)}=1,\\
&{\vs{\omega}}^{(m)\ddagger}\cdot\vv{u}^{(m)*}+v^{(m)*}_m=0,
\end{split}
\end{equation}
where `$\cdot$' denotes the scalar product of two vectors. The 
\emph{inner bordering} scheme can now be stated as 
follows: 
\begin{itemize}
\item For $m = 0$, set
\begin{equation}
\begin{pmatrix}
\vv{u}^{(0)}\\
\vv{v}^{(0)}
\end{pmatrix} = \frac{1}{(1+|{\omega}_0|^2)}
\begin{pmatrix}
1\\
-{\omega}_0
\end{pmatrix}.
\end{equation}
\item Introduce $c_m$ and $d_m$, such that
\begin{equation}\label{eq:inner-bordering-vec}
\begin{split}
&\vv{u}^{(m+1)} = 
c_m\begin{pmatrix}
\vv{u}^{(m)}\\
0
\end{pmatrix}+
d^*_m\begin{pmatrix}
0\\
-\vv{v}^{(m)\ddagger}
\end{pmatrix},\\
&\vv{v}^{(m+1)} = 
c_m\begin{pmatrix}
\vv{v}^{(m)}\\
{0}
\end{pmatrix}+
d^*_m\begin{pmatrix}
{0}\\
\vv{u}^{(m)\ddagger}
\end{pmatrix},
\end{split}
\end{equation}
where $c_m$ and $d_m$ are determined such that~\eqref{eq:inverse-toeplitz} holds
for $m+1$. This yields $c_m = (1+|\beta_m|^2)^{-1}$ and $d_m = -\beta_m c_m$
where
\begin{equation}
\begin{split}
\beta_m
&=\sum_{j=0}^{m}{\omega}^*_{m+1-j}u^{(m)*}_j\\
&={\vs{\omega}}^{(m+1)*}\cdot
\begin{pmatrix}
{0}\\
\vv{u}^{(m)\ddagger}
\end{pmatrix}
={\vs{\omega}}^{(m+1)\ddagger}\cdot
\begin{pmatrix}
\vv{u}^{(m)*}\\
{0}
\end{pmatrix}.
\end{split}
\end{equation}
\end{itemize}
The recurrence relation for $u_0^{(m)}$ and $v_0^{(m)}$ can be solved
explicitly: From $u_0^{(m+1)} = c_m u_0^{(m)}$, we have
\begin{equation}\label{eq:CQ-u0}
u_0^{(m+1)}=c_m c_{m-1}\ldots c_0;
\end{equation}
and from $v_0^{(m+1)} = c_m v_0^{(m)}$, we have
\begin{equation}\label{eq:CQ-v0}
\begin{split}
v_0^{(m+1)} 
&=-c_m c_{m-1}\ldots c_0\frac{\omega_0}{1+|\omega_0|^2}\\
&= -\frac{\omega_0}{1+|\omega_0|^2}u_0^{(m+1)}.
\end{split}
\end{equation}
Now, in order to determine the components 
$\chi^{(m+1)}_{0}$ and $\chi^{(m+1)}_{2m+3}$, based on the discussion above, it
is straightforward to setup the following system of equations
\begin{equation}\label{eq:chi-first-and-last}
\begin{split}
&\chi^{(m+1)}_{0}=\begin{pmatrix}
\vv{u}^{(m+1)}\\
-\vv{v}^{(m+1)}
\end{pmatrix}^{*}\cdot
\begin{pmatrix}
-h{\chi}^{(m+1)}_{2m+3}\breve{\vs{p}}^{(m+1)\ddagger}\\
-\breve{\vv{p}}^{(m+1)}+h{\chi}^{(m+1)}_{0}\breve{{\vs{p}}}^{(m+1)}
\end{pmatrix},\\
&\chi^{(m+1)}_{2m+3}=\begin{pmatrix}
\vv{v}^{(m+1)\ddagger}\\
\vv{u}^{(m+1)\ddagger}
\end{pmatrix}^{*}\cdot
\begin{pmatrix}
-h{\chi}^{(m+1)}_{2m+3}\breve{\vs{p}}^{(m+1)\ddagger}\\
-\breve{\vv{p}}^{(m+1)}+h{\chi}^{(m+1)}_{0}\breve{{\vs{p}}}^{(m+1)}
\end{pmatrix}.
\end{split}
\end{equation}
In order to simplify this linear system above, we compute the scalar
products using the relations~\eqref{eq:inner-bordering-basic} 
and~\eqref{eq:inner-bordering-vec} as follows:
\begin{equation}
\begin{split}
2h\vv{u}^{(m+1)*}\cdot\breve{\vs{p}}^{(m+1)\ddagger} 
&= c_m\beta_m-2hd_m\vv{v}^{(m)}\cdot\breve{\vs{p}}^{(m)*}\\
&= c_m\beta_m+d_m(1-u_0^{(m)}-{\omega}^*_0v_0^{(m)})\\
&=\frac{c_mu_0^{(m)}}{1+|\omega_0|^2}\beta_m.
\end{split}
\end{equation}
Set $\lambda_m\equiv{u_0^{(m)}}/{(1+|\omega_0|^2)}\in\field{R}_+$. Next, we have
\begin{equation}
\begin{split}
2h\vv{v}^{(m+1)}\cdot\breve{{\vs{p}}}^{(m+1)*}
&=2hc_m\vv{v}^{(m)}\cdot\breve{{\vs{p}}}^{(m)*}+d^*_m\beta_m,\\
&=\frac{c_mu_0^{(m)}}{1+|\omega_0|^2}-1=\lambda_{m+1}-1.
\end{split}
\end{equation}
Using these results, the linear system~\eqref{eq:chi-first-and-last} can be stated
as
\begin{equation}
\begin{split}
&(1+\lambda_{m+1})\chi_0^{(m+1)}
+\beta_m\lambda_{m+1}\chi_{2m+3}^{(m+1)}
=-\frac{1}{h}(1-\lambda_{m+1}),\\
&\beta^*_m\lambda_{m+1}\chi_{0}^{(m+1)}
-(1+\lambda_{m+1})\chi_{2m+3}^{(m+1)}
=\frac{1}{h}\lambda_{m+1}\beta^*_m.
\end{split}
\end{equation}
Note that if $\breve{\rho}(\zeta)$ decays at 
least as $\bigO{\zeta^{-2}}$, then $\breve{p}_0=0$ and $\omega_0=0$. 
Solving for the general case, we have
\begin{equation}
\int_{0}^{t_{m+1}}|q(T_+-s)|^2ds
\approx-2\chi_0^{(m+1)}
=\frac{2}{h\Delta_m}\left(1-c_m\lambda_{m}^2\right),
\end{equation}
and
\begin{equation}\label{eq:q-TIB}
q(T_+-t_{m+1})
\approx 2\chi_{2m+3}^{(m+1)*}=-\frac{4\beta_mc_m\lambda_{m}}{h\Delta_m},
\end{equation}
where
\begin{equation}
\begin{split}
\Delta_m 
&=\left[1+2c_m\lambda_{m}
+c_m\lambda_{m}^2\right]\\
&=\left[1+2c_mu_0^{(m)}+c_m(u_0^{(m)})^2\right]+\bigO{h^2}.
\end{split}
\end{equation}
The algorithm proposed by Belai~\et~\cite{BFPS2007,FBPS2015} uses the 
approximations $c_m\approx1$ and $u_0^{(m)}\approx1$ in~\eqref{eq:q-TIB} while maintaining
the accuracy of $\bigO{h^2}$; however, this step may lead to higher truncation 
error on account of the relation~\eqref{eq:CQ-u0} since
\begin{equation}
    u_0^{(m)}=1-\sum_{k=0}^{m-1}|\beta_m|^2+\ldots.
\end{equation}

\begin{figure*}[!tbh]
\centering
\includegraphics[scale=1]{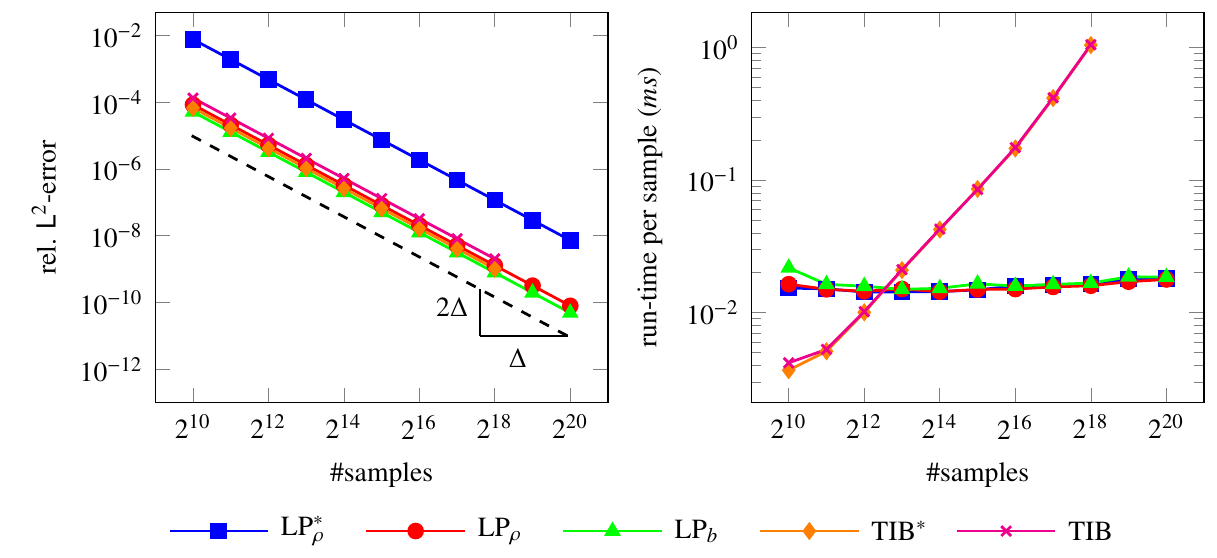}
\caption{\label{fig:LP-TIB-comp}The figure shows a comparison of the 
algorithms  LP$^*_{\rho}$, LP$_{\rho}$, LP$_b$, TIB$^*$ and TIB for the secant-hyperbolic
potential ($A_R=0.4$) with respect to convergence rate (left) and run-time per sample (right).}
\end{figure*}

\begin{figure*}[!th]
\centering
\def\scale{1}
\includegraphics[scale=\scale]{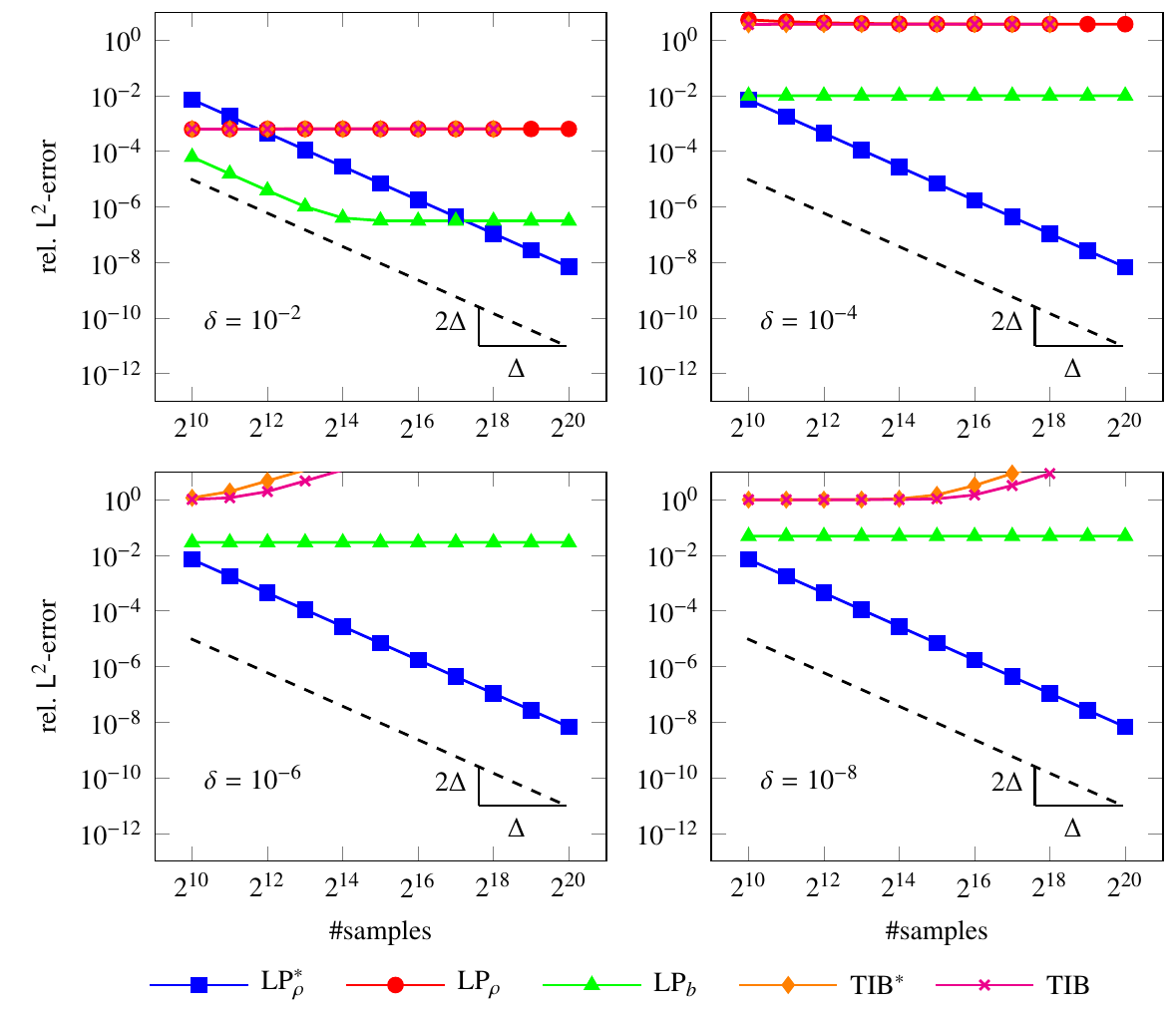}
\caption{\label{fig:Lubich} The figure shows the convergence behavior of the 
algorithms  LP$^*_{\rho}$, LP$_{\rho}$, LP$_b$, TIB$^*$ and TIB for the secant-hyperbolic
potential with $A_R=0.5-\delta$.}
\end{figure*}

\subsection{Discrete inverse scattering with the b-coefficient}
\label{sec:lp-b}
Let the $b$-coefficient be of exponential-type such that its Fourier-Laplace
transform, $\beta(\tau)$, is supported in $[-2T_+,2T_-]$. Let
$\breve{b}(\zeta)=b(\zeta)e^{2i\zeta T_+}$ so that
\begin{equation}\label{eq:integ-b}
\breve{b}(\zeta)=\int^{2W}_{0}\breve{\beta}(\tau)e^{i\zeta\tau}d\tau,
\end{equation}
where $W=T_++T_-$ and $\breve{\beta}(\tau)$ is the shifted version of
$\beta(\tau)$ with support in $[0,2W]$. We assume that 
$\breve{\beta}(\tau)$ has derivatives of order $>2$ in
$(0,2W)$ and they vanish at the end points so that the integral can be
approximated by the sum
\begin{equation}
\breve{b}(\zeta) 
\approx 2h\sum_{k=0}^{N-1}\breve{\beta}(2hk) [e^{2i\zeta h}]^k
=2h\sum_{k=0}^{N-1}\breve{\beta}(2hk)z^{2k}.
\end{equation}
If $\breve{\beta}(\tau)$ is smooth, then the approximation above achieves spectral
accuracy otherwise it is at least $\bigO{h^2}$ accurate.
Again, if $\breve{\beta}(\tau)$ cannot be evaluated exactly, we use 
the FFT algorithm to compute the integral
\begin{equation}\label{eq:integ-beta}
\breve{\beta}(\tau)=\frac{1}{2\pi}\int^{\infty}_{-\infty}\breve{b}(\xi)e^{-i\xi\tau}d\xi.
\end{equation}
The decay property of
$\breve{b}(\xi),\,\xi\in\field{R},$ decides 
the convergence of the discrete Fourier sum used in FFT so that it may 
be necessary to employ large number of samples in order to get
accurate results. 

The next step consists of constructing 
a polynomial approximation for
$a(\zeta)$ in $|z|<1$ (under the assumption that no bound
states are present). With a slight abuse of notation, let us denote the
polynomial approximation for $\breve{b}(\zeta)$ by $\breve{b}_N(z^2)$ or 
simply $\breve{b}(z^2)$\footnote{Note that the mapping $z=e^{i\xi h}$ plays no
role once $\breve{b}_N(z^2)$ is determined so that one can simply speak of the
variable $z$.}. Here, the relation~\cite{AKNS1974,AS1981}
$|a(\xi)|^2+|\breve{b}(\xi)|^2=1,\,\xi\in\field{R},$ allows us to set up a 
Riemann-Hilbert (RH) problem for a sectionally analytic function
\begin{equation}
\tilde{g}(z^2)=\begin{cases}
g(z^2) & |z|<1,\\
-{g}^*(1/z^{*2}) & |z|>1,
\end{cases}
\end{equation}
such that the jump condition is given by
\begin{equation}\label{eq:RH-circ}
\tilde{g}^{(-)}(z^2) - \tilde{g}^{(+)}(z^2) =
\log[1-|\breve{b}(z^2)|^2],\quad |z|=1,
\end{equation}
where $\tilde{g}^{(-)}(z^2)$ and $\tilde{g}^{(+)}(z^2)$ denotes the boundary values when
approaching the unit circle from $|z|<1$ and $|z|>1$, respectively. 
Note that for the jump function to be well-defined, we must have
$|\breve{b}(z^2)|<1$. Let the jump function on the right hand side of~\eqref{eq:RH-circ} be denoted by 
$f(z^2)$ which can be expanded as a Fourier series
\begin{equation}\label{eq:f-series}
f(z^2)=\sum_{k\in\field{Z}}f_kz^{2k},\quad |z|=1.
\end{equation}
Now, the solution to the RH problem can be stated using the Cauchy integral~\cite[Chap.~14]{Henrici1993} 
\begin{equation}
\tilde{g}(z^2) = \frac{1}{2\pi i}\oint_{|w|=1}\frac{f(w)}{z^2-w}dw.
\end{equation}
The function $g(z^2)$ analytic in $|z|<1$ then works out to be
\begin{equation}
g(z^2)=\sum_{k\in\field{Z}_+\cup\{0\}}f_kz^{2k},\quad |z|<1.
\end{equation}
Finally, $a_N(z^2)=\{\exp[g(z^2)]\}_{N}$ with $z=e^{i\zeta h}$ where
$\{\cdot\}_N$ denotes truncation after $N$ terms. The implementation of the 
procedure laid out above can be carried out using the FFT
algorithm, which involves computation of the coefficients $f_k$ and the 
exponentiation in the last step. Note that, in the computation of
$g(z^2)$, we discarded half of the coefficients; therefore, in the numerical implementation
it is necessary to work with at least $2N$ number of samples of
$f(z^2)$ in order to obtain $a_N(z^2)$ which is a polynomial of degree $N-1$.

Finally, the input to the layer-peeling algorithm can 
be taken to be
\begin{equation*}
P^{(N)}_1(z^2) = a_N(z^2),\quad
P^{(N)}_2(z^2) =\breve{b}_N(z^2).
\end{equation*}
For the TIB algorithm, we compute $\breve{\rho}_N=\breve{b}_N/a_N$ so that the discrete 
approximation to $p(\tau)$, which serves as the input, can be computed.

\section{Numerical Results}
\label{sec:results}
Our objective in this section is twofold: First to test the accuracy as well as complexity of the
proposed methods, second, to numerically verify some of the analytical results obtained thus
far. If the analytic solution is known, we quantify the error as 
\begin{equation}\label{eq:e_rel-q}
e_{\text{rel.}}={\|q^{(\text{num.})}-q\|_{\fs{L}^2}}/{\|q\|_{\fs{L}^2}},
\end{equation}
where $q^{(\text{num}.)}$ denotes the numerically computed potential and $q$ is
the exact potential. The integrals are evaluated numerically using the trapezoidal 
rule. When the exact solution is not known, a higher-order 
scheme such as the (exponential) $3$-step \emph{implicit Adams} 
method (IA$_3$)~\cite{V2018LPT,V2017BL}, which has an order of convergence $4$, i.e.,
$\bigO{N^{-4}}$, can be used to compute the accuracy of the numerical solution
by computing the nonlinear Fourier spectrum using $q^{(\text{num}.)}$ as the
scattering potential and compare it with the known nonlinear Fourier spectrum. The error 
computed in this manner does not represent the \emph{true} numerical error; therefore, the results in this 
case must be interpreted with caution.

\subsection{Truncated secant-hyperbolic potential}
\label{app:sech-trunc}
In order to devise a test of convergence for one-sided potentials, we derive
the scattering coefficients for a truncated secant-hyperbolic potential. To this
end, consider $q(t) = A \sech t$
where $A<0.5$ so that there are no bound states present. Let $s=\tfrac{1}{2}(1-\tanh
t)$, then the Jost solution $\vs{\phi}(t;\zeta)$ is given by~\cite{SY1974}
\begin{align*}
&\phi_1 = e^{-i\zeta t}{}_2F_1\biggl(A,-A;\frac{1}{2}-i\zeta;1-s\biggl),\\
&\phi_2 = -\frac{Ae^{-i\zeta t}}{\frac{1}{2}-i\zeta}
\left(\frac{1}{2}\sech t\right)
{}_2F_1\left(1-A,1+A;\frac{3}{2}-i\zeta;1-s\right),
\end{align*}
where ${}_2F_1$ denotes the hypergeometric
function~\cite{Olver:2010:NHMF}. The scattering coefficients are given by
\begin{align*}
a(\zeta) 
         &=\frac{\Gamma\left(\frac{1}{2}-i\zeta\right)^2}
    {\Gamma\left(A+\frac{1}{2}-i\zeta\right)\Gamma\left(-A+\frac{1}{2}-i\zeta\right)},\quad\Im{\zeta}\geq0,\\
b(\zeta) &= -\frac{\sin(\pi A)}{\cosh(\pi\zeta)},\quad|\Im{\zeta}|<\frac{1}{2}. 
\end{align*}
Let $(-\infty,\,T]$ be the computational domain so that the truncation takes place at 
$t=T$ (with $s(T)=s_2$), the scattering coefficients of the left-sided potential are given by
\begin{align*}
&a^{(-)}(\zeta) 
= {}_2F_1\left(A,-A;\frac{1}{2}-i\zeta;1-s_2\right),\\
&\breve{b}^{(-)}(\zeta) 
= -\frac{A\sech
T}{2\left(\frac{1}{2}-i\zeta\right)}\,\cdot\,{}_2F_1\left(1-A,1+A;\frac{3}{2}-i\zeta;1-s_2\right)
\end{align*}
Note that $s(T)\approx e^{-2T}$, therefore, we may use the asymptotic form
of the hypergeometric functions to obtain the scattering coefficients. For
$a$-coefficient, the approximation works out to be
$a^{(-)}(\zeta)=a(\zeta)+\bigO{e^{-T(2\Im(\zeta)+1)}}$.
For the $b$-coefficient, we observe that 
\begin{equation}\label{eq:truncated-b}
\breve{b}^{(-)}(\zeta)
\sim Ae^{-T}\frac{a(\zeta)}{\frac{1}{2}+i\zeta} 
-\frac{\sin\pi A}{\cosh\pi\zeta}
e^{2i\zeta T}+\bigO{e^{-2T}},
\end{equation}
for $\zeta\neq i(n+\frac{1}{2}),\,n=0,1,\ldots$,
and otherwise
\begin{equation}
\breve{b}^{(-)}(\zeta)
\sim\begin{cases}
    -\frac{2T}{\pi}e^{-T}\sin\pi A\, &\zeta=i/2,\\
Ae^{-T}\frac{a(\zeta)}{\frac{1}{2}+i\zeta},&\zeta=3i/2,\,5i/2,\ldots\,
\end{cases}
\end{equation}
in the upper half of the complex plane. For $\Im{\zeta}=\eta>1$, the damping
term in~\eqref{eq:truncated-b}, $e^{2i\zeta T}=\bigO{e^{-2T}}$; therefore
\begin{equation}
\breve{b}^{(-)}(\zeta)
\sim
Ae^{-T}\frac{a(\zeta)}{\frac{1}{2}+i\zeta}+\bigO{e^{-2T}},\quad\Im{\zeta}>1.
\end{equation}
Let $\tilde{A}=A+1/2$. For large $\zeta$, we use the following asymptotic forms in order to evaluate
$a(\zeta)$~\cite{Olver:2010:NHMF}:
\begin{equation*}
\begin{split}
a(\zeta)&=\frac{\Gamma\left(-i\zeta+\frac{1}{2}\right)}{\Gamma\left(-i\zeta+\tilde{A}\right)}
\left[\frac{\Gamma\left(-i\zeta + 1-\tilde{A}\right)}
{\Gamma\left(-i\zeta+\frac{1}{2}\right)}\right]^{-1}\sim\left(\frac{z_2}{z_1}\right)^{\tilde{A}-1/2}\\
&\times\left[1+\frac{H_1(1/2,\tilde{A})}{z_1^2}+\frac{H_2(1/2,\tilde{A})}{z_1^4}
+\ldots\right]\\
&\times \left[1+\frac{H_1(1-\tilde{A},1/2)}{z_2^2}+\frac{H_2(1-\tilde{A}, 1/2)}{z_2^4}
    +\ldots\right]^{-1}
\end{split}
\end{equation*}
where $z_1 = -i\zeta+(2\tilde{A}-1)/4$ and $z_2 = -i\zeta+(1-2\tilde{A})/4$.
The other constants in the expansion are defined as
\begin{align*}
&H_1(r,s)= -\frac{1}{12}\binom{r-s}{2}(r-s+1),\\
&H_2(r,s)= \frac{1}{240}\binom{r-s}{4}[2(r-s+1)+5(r-s+1)^2].
\end{align*}
The truncation point $t=T$ is chosen large enough so that the reflection
coefficient $\rho^{(-)}(\xi)$ becomes effectively a Schwartz class function. This would
allow us to test for the following algorithms:
\begin{itemize}
\item LP$^*_{\rho}$: This refers to the layer-peeling algorithm where the input
is synthesized using complex integration in the upper half of the complex plane
as described in Sec.~\ref{sec:lp-rho}. For simplicity, we may also refer to it as
the Lubich method.
\item LP$_{\rho}$: This refers to the layer-peeling algorithm where the input
is synthesized using the discrete approximation of the Fourier integral as 
Sec.~\ref{sec:lp-rho}.
\item LP$_{b}$: This refers to the layer-peeling algorithm where the input
is synthesized using the $b$-coefficient as described in Sec.~\ref{sec:lp-b}.
\item TIB: This refers to the second order TIB algorithm reported in~\cite{BFPS2007,FBPS2015}.
\item TIB$^*$: This refers to the modified second order TIB algorithm presented
in Sec.~\ref{sec:modified-TIB}.
\end{itemize}
The over sampling factor for synthesizing the input is taken to be
$n_{\text{os}}=8$. This means that for Cauchy or the Fourier integral, the
number of samples of $\rho^{(-)}(\zeta)$ used is $M=8N$ where 
$N\in\{2^{10},2^{11},\ldots,2^{20}\}$.

For the first test, we set $A_R=0.4$. The results of the convergence and
complexity analysis are plotted in Fig.~\ref{fig:LP-TIB-comp}. The results show 
that all the methods exhibit a second order of convergence. The Lubich method
where complex integration is employed turns out to be the least accurate. The TIB
and TIB$^*$ show small difference in accuracy. The plot depicting the total
run-time per sample shows linear growth for TIB and TIB$^*$ with respect to
$\log N$ while the others exhibit a logarithmic growth. Therefore, the
layer-peeling algorithm is an order of magnitude faster than the integral method.

For the second test, we set $A_R=0.5-\delta$ where $\delta\in\{10^{-2},
10^{-4},10^{-6},10^{-8}\}$. From the analytic solution, it is evident that as the
value of $\delta$ decreases, a pole in $\field{C}_-$ approaches the real axis.
This makes it a challenging test case for methods that sample $\rho^{(-)}(\zeta)$ on
the real axis on account of the numerical ill-conditioning introduced by the
proximity of the pole to the real axis. The results of the convergence analysis are 
plotted in Fig.~\ref{fig:Lubich} where all methods tend to fail except for 
LP$^*_{\rho}$, the Lubich method. The unusual stability of this method is attributed 
to the fact that samples of $\rho^{(-)}(\zeta)$ is sought in $\field{C}_+$ 
avoiding any ill-conditioning introduced as a result of
the proximity of the pole in $\field{C}_-$ to the real axis.

\begin{figure*}[!t]
\centering
\def\scale{1}
\includegraphics[scale=\scale]{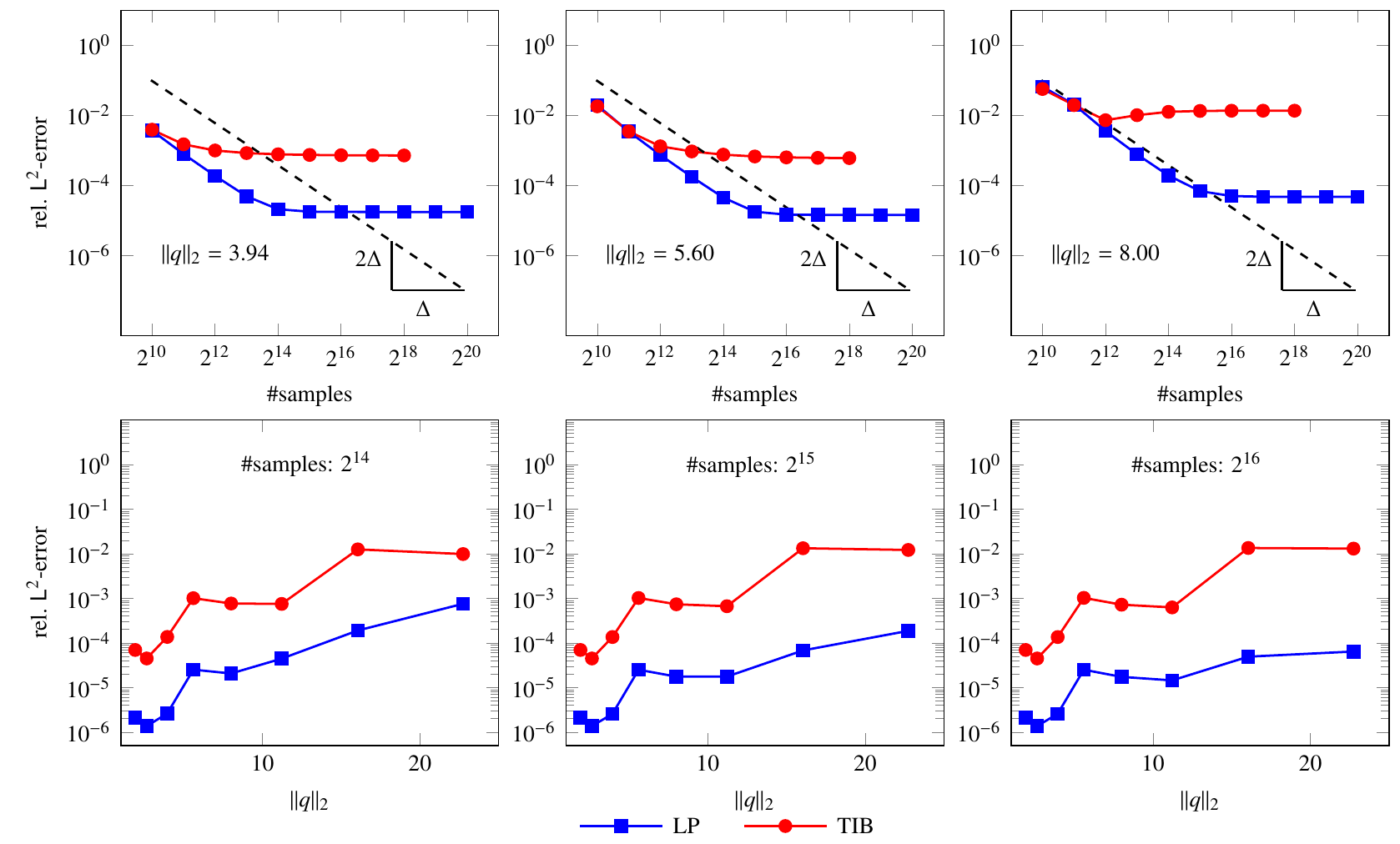}
\caption{\label{fig:convg-rc} The figure shows the error
analysis for the signal generated from the $b$-coefficient given 
by~\eqref{eq:RC-spec}. The error is quantified by~\eqref{eq:e_rel-rho}.}
\end{figure*}

\begin{figure*}[!t]
\centering
\includegraphics[scale=1]{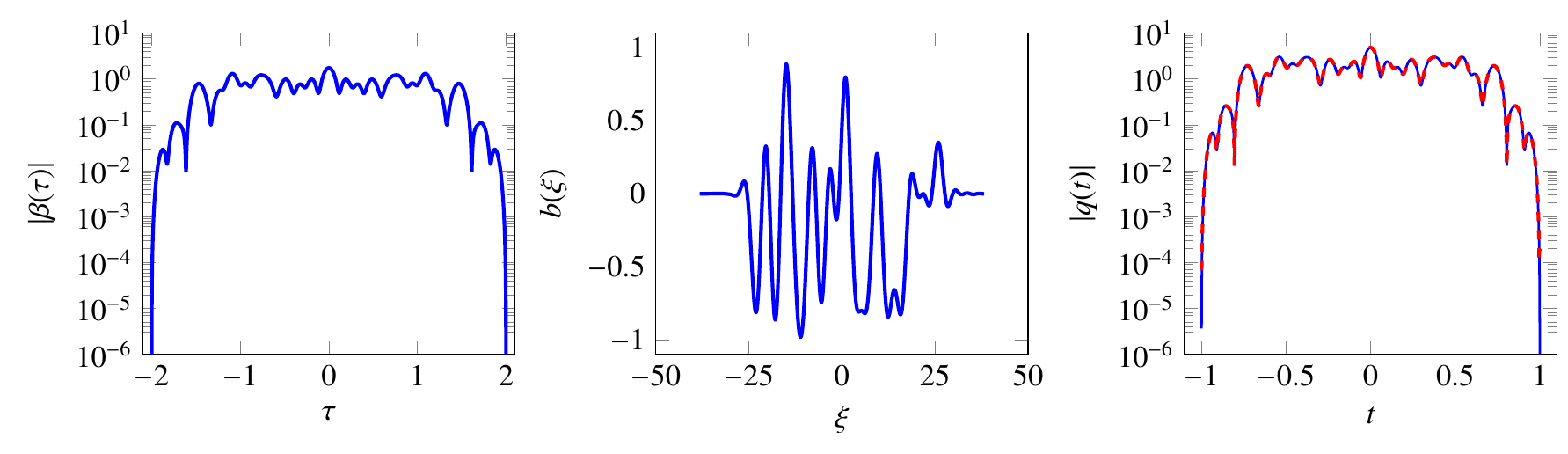}
\caption{\label{fig:ex_rc32}
The figure shows an example of synthesized
time-limited signal from a compactly supported $\beta(\tau)$ or a
$b$-coefficient that is of exponential type. Here $N_{\text{sym}}=32$ and
$\|q\|_2=2.64$. In the last plot the 
solid and the dashed lines correspond to the potential computed
using the LP and the TIB algorithm, respectively.%
}
\end{figure*}

\begin{figure*}[!t]
\centering
\includegraphics[scale=1]{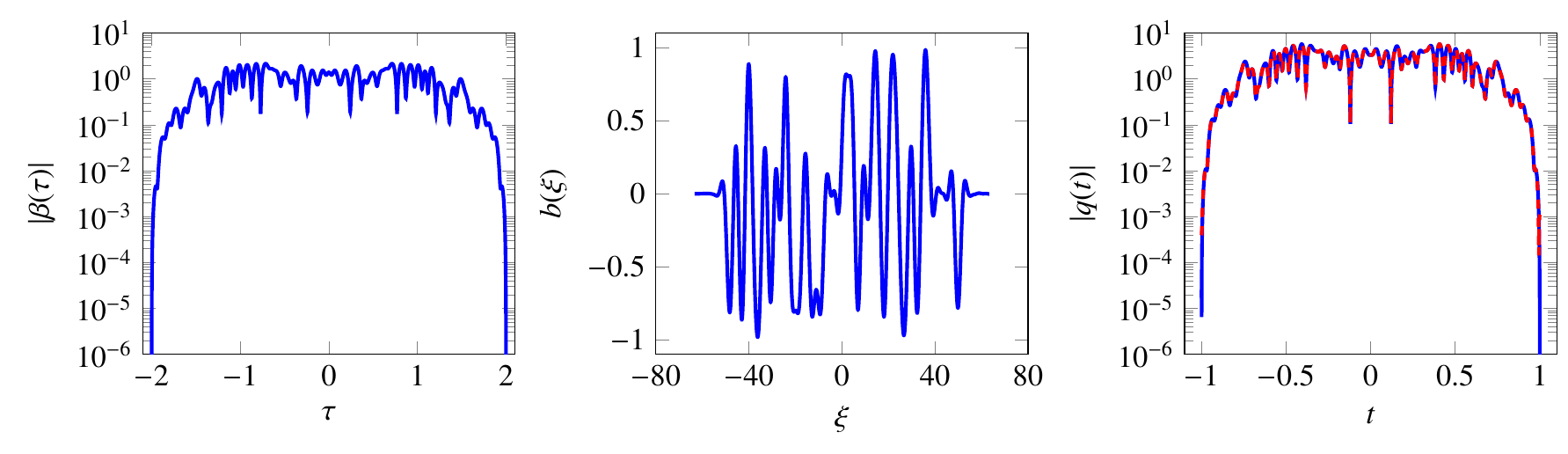}
\caption{\label{fig:ex_rc64}
The figure shows an example of synthesized
time-limited signal from a compactly supported $\beta(\tau)$ or a
$b$-coefficient that is of exponential type. Here $N_{\text{sym}}=64$ and
$\|q\|_2=3.94$. In the last plot the solid and the dashed lines correspond 
to the potential computed using the LP and the TIB algorithm, respectively.%
}
\end{figure*}

\begin{figure*}[!t]
\centering
\includegraphics[scale=1]{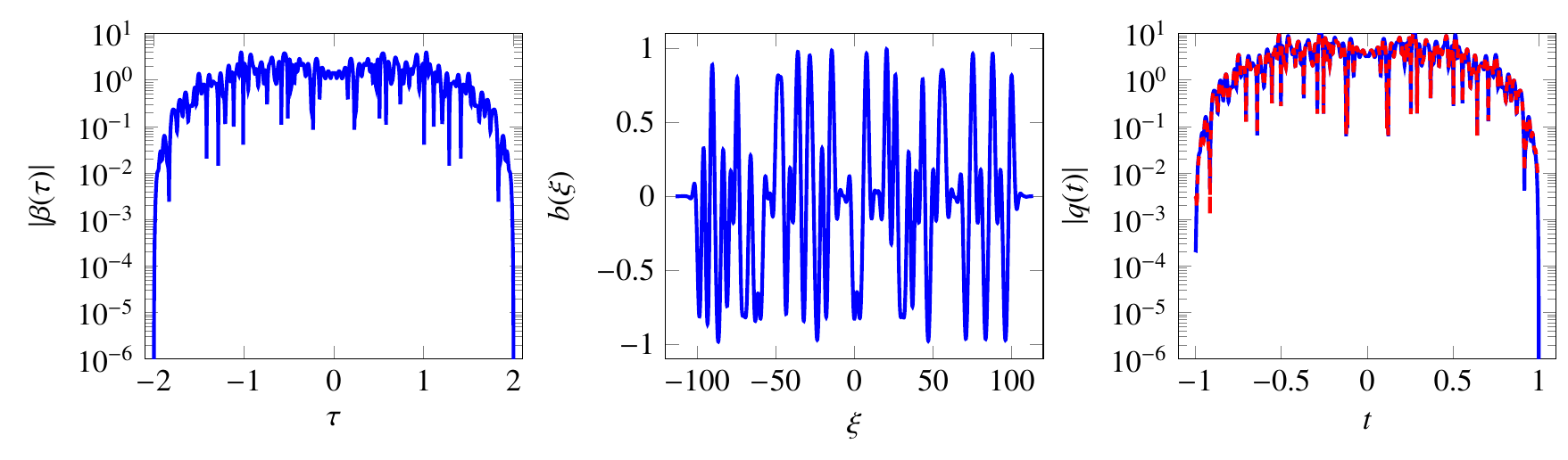}
\caption{\label{fig:ex_rc128}
The figure shows an example of synthesized
time-limited signal from a compactly supported $\beta(\tau)$ or a
$b$-coefficient that is of exponential type. Here $N_{\text{sym}}=128$ and
$\|q\|_2=5.6$. In the last plot the solid and the dashed lines correspond to 
the potential computed using the LP and the TIB algorithm, respectively.%
}
\end{figure*}

\subsection{Raised-cosine filter}
In this test case, we consider generation of time-limited signals using a
$b$-coefficient that is an entire function of exponential type. Let 
$\Omega_h = [-\pi/2h,\pi/2h]$, then the error is quantified by 
\begin{equation}\label{eq:e_rel-rho}
e_{\text{rel.}}=
{\|\rho^{(\text{num.})}-\rho\|_{\fs{L}^2(\Omega_h)}}/{\|\rho\|_{\fs{L}^2(\Omega_h)}},
\end{equation}
where the integrals are computed from $N$ equispaced samples in $\Omega_h$ using the
trapezoidal rule. The quantity $\rho^{(\text{num.})}$ is computed using the (exponential) IA$_3$.

Let us define
\begin{equation}
f_{\text{rc}}(\xi) = A_{\text{rc}}\sinc\left[2(1-\chi)\xi T\right]
\frac{\cos\left(2\chi\xi T\right)}{1-\left(\frac{4\chi\xi T}{\pi}\right)^2},
\end{equation}
which is a representation of the raised-cosine filter (in the ``time'' domain).
We set $A_{\text{rc}}=0.489$, $\chi=0.3$ and $T=1$ so that the support of the potential works out to be
$[-T,T]$. We then synthesize $b$ as follows:
\begin{equation}\label{eq:RC-spec}
b(\xi)=\sum_{n\in J}c_nf_{\text{rc}}(\xi-\xi_n)
\end{equation}
where $c_n\in\{-1,+1\}$ is chosen randomly and $\xi_n=(2\pi/4T)n$. The index set
is defined as $J=\{-N_{\text{sym}}/2,\ldots,N_{\text{sym}}/2-1\}$ 
with $N_{\text{sym}}\in\{16,32,\ldots,2048\}$. Here, we would conduct numerical
experiments with varying number of samples and fixed $N_{\text{sym}}$ or with
different values $N_{\text{sym}}$ and fixed $N$. In these experiments, we compare the
LP algorithm to the TIB algorithm presented in Sec.~\ref{sec:modified-TIB}.
The input to both of these algorithms is synthesized from the $b$-coefficient as discussed in
Sec.~\ref{sec:lp-b}.

The results of the error analysis are displayed in Fig.~\ref{fig:convg-rc}. The
plots in the top row indicate that the error as a function of $N$ shows a second
order slope before it plateaus after a certain value. This plateauing does 
\emph{not} necessarily indicate the lack of convergence because the error does not represent 
the true numerical error. Further, it is evident from these plots that the error for TIB
algorithm plateaus much earlier than that of the LP algorithm; therefore, the LP
algorithm is superior to TIB in terms of accuracy. The same conclusion about the
accuracy of LP as compared to TIB can be drawn from the plots in the bottom row 
where the error is plotted as a function of the $\fs{L}^2$-norm of the potential 
which is the result of varying $N_{\text{sym}}$.

The specific examples of the signals generated with $N=2^{12}$ for $N_{\text{sym}}=32,64,128$
are displayed in figures~\ref{fig:ex_rc32},~\ref{fig:ex_rc64}
and~\ref{fig:ex_rc128}, respectively. The generated signal shows the correct
decay behavior such that the compact support of the potential can be inferred
easily.

\begin{figure}[!t]
\centering
\includegraphics[scale=1]{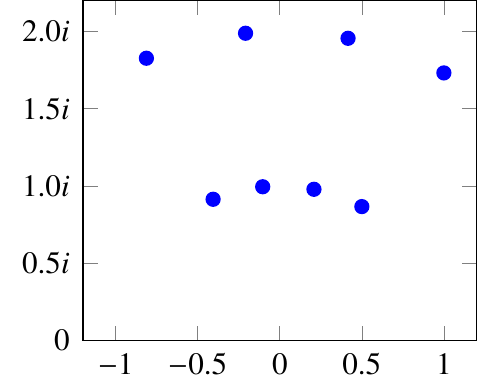}
\caption{\label{fig:eigs}The figure shows the eigenvalues to be `added' to the
compactly supported potential generated using the $b$-coefficient.}
\end{figure}

\begin{figure*}[!t]
\centering
\subfloat[$N=2^{12}$]{\includegraphics[scale=1]{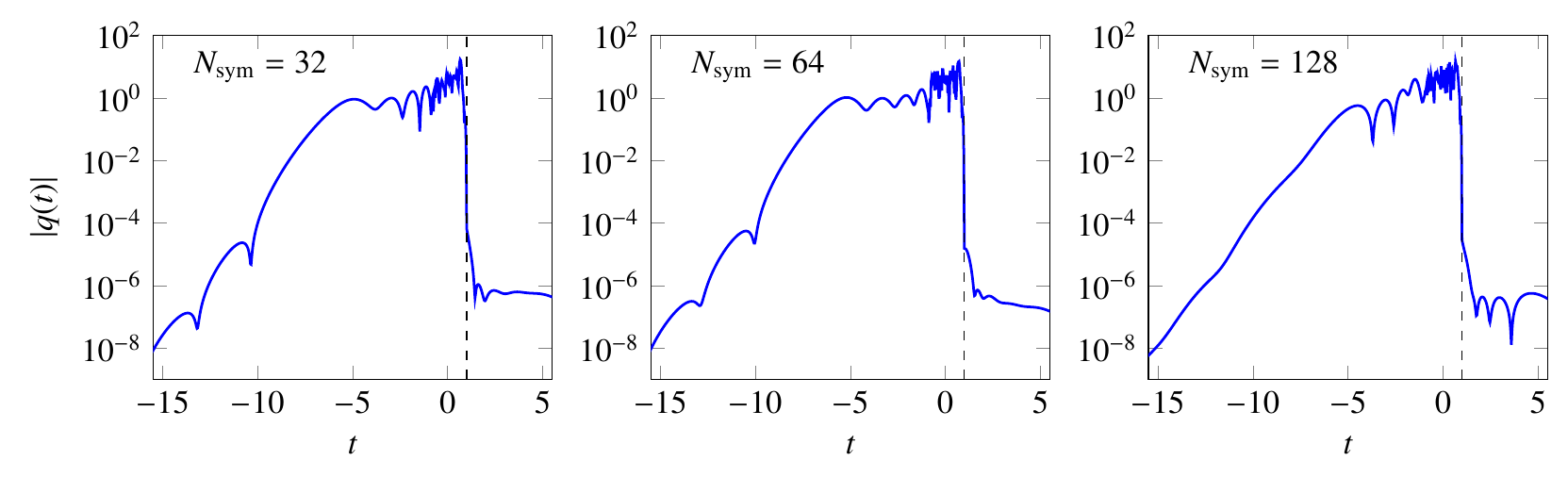}}\\
\subfloat[$N=2^{14}$]{\includegraphics[scale=1]{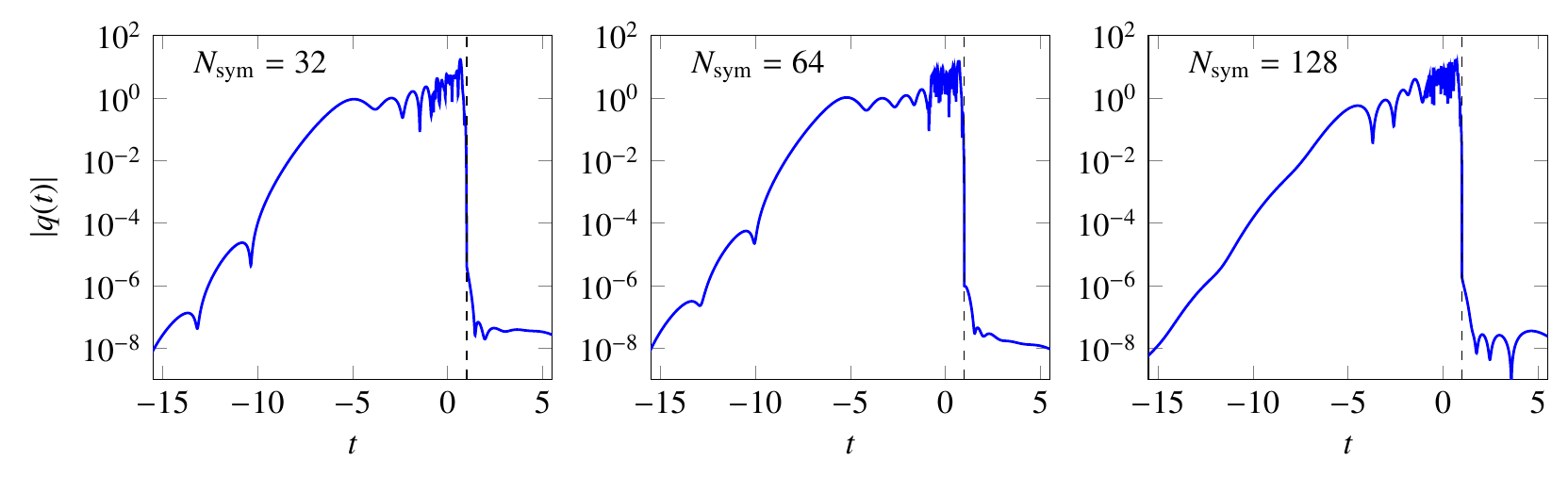}}
\caption{\label{fig:bsex_rc}The figure shows the signal generated using the
$b$-coefficient given by~\eqref{eq:RC-spec} and the eigenvalues shown in
Fig.~\ref{fig:eigs}. The norming constants are chosen such that the signal is
supported in $(-\infty,1]$. It is evident that increasing the number of samples
($N$) tends to make the
tail beyond $t=1$ (marked by the vertical dashed line) less significant.}
\end{figure*}

The final test that we consider has to do with the addition of bound states to 
compactly supported potentials. The signal is generated using the fast Darboux
transformation algorithm (FDT) reported in~\cite{V2017INFT1}. Define
\begin{equation}
\vs{\theta} =(\pi/3, 13\pi/30, 8\pi/15, 19\pi/30)
\end{equation}
and let the set of eigenvalues be $\{\exp(i\vs{\theta}),2\exp(i\vs{\theta})\}$
(see Fig.~\ref{fig:eigs}). The norming constants are chosen such that the
resulting potential is supported on $(-\infty,1]$, i.e., $b_k=b(\zeta_k)$. The 
specific examples of the signals are generated with 
$N=2^{12},2^{14}$ for $N_{\text{sym}}=32,64,128$ are displayed in 
Fig.~\ref{fig:bsex_rc}. The generated signal shows a non-zero
tail beyond $t=1$ which can be attributed to the numerical errors. The
tail, however, becomes less significant as the number of samples is increased
from $N=2^{12}$ to $N=2^{14}$.

\section{Conclusion}\label{sec:conclusion}
To conclude, we have presented some useful regularity properties of the nonlinear
Fourier spectrum under various assumptions on the scattering potential. These results 
bear a close resemblance to the Paley-Wiener theorems for Fourier-Laplace transform. We 
have further provided the necessary and sufficient 
conditions for the nonlinear Fourier spectrum such that the corresponding
scattering potential has a prescribed support. In addition, we also analyzed
how to exploit the regularity properties of the scattering coefficients to
improve the numerical conditioning of the differential approach of inverse
scattering. For the sake of comparison, we considered the integral approach 
originally presented in~\cite{BFPS2007,FBPS2015} and provided the correct 
derivation of the so called T\"optlitz inner-bordering scheme. The comparative 
study of the two approaches conduted in the article clearly reveals that the 
differential approach is more accurate while admittedly
being faster by an order of magnitude for the class of nonlinear Fourier
spectra considered in this article.

\appendix
\section{Proof of Lemma~\ref{lemma:RL-compact}}\label{app:RL-compact}
The proof is almost similar to that of the Riemann-Lebesgue 
lemma~\cite[Chap.~13]{Jones2001}. First, we note that
\[
|\Phi_2(t;\zeta)|\leq\int_{-T_-}^{t}|q(y)|e^{-2\eta(t-y)}dy\leq
\max\left(1,e^{-2\eta W}\right)\|q\|_{\fs{L}^1},
\]
where $W=T_-+T_+$. Let $\tau=\pi\xi/(2|\xi|^2)$ and consider the case
$\eta=0$: From the relation
\begin{align*}
&\int_{-\infty}^{t}r(y+\tau)e^{2i\xi(t-y)}dy=-\int_{-\infty}^{t+\tau}r(u)e^{2i\xi(t-u)}du\\
&=-\Phi_2(t;\xi)-\int_{t}^{t+\tau}r(u)e^{2i\xi(t-u)}du,
\end{align*}
we may write
\begin{equation}\label{eq:phi-estimate}
2|\Phi_2(t;\xi)|\leq\int_{-\infty}^{T_+}|q(y)-q(y+\tau)|dy\\+\int_{t}^{t+\tau}|q(y)|dy.
\end{equation}
Now, as $|\xi|\rightarrow\infty$, we
have $\tau\rightarrow0$ so that, using the continuity of translation for
$\fs{L}^1$-functions~\cite[Chap.~7]{Jones2001}, we conclude that the first term
on the right hand side of above tends to zero. The second term can be shown to approach the
limit $0$ uniformly as $\tau\rightarrow0$ on account of the theorem on absolute
continuity~\cite[Chap.~6]{Jones2001}. Therefore, we conclude $|\Phi_2(t;\xi)|\rightarrow0$ 
uniformly in $\Omega$ as $|\xi|\rightarrow\infty$. For fixed $\eta>0$, we have
\begin{multline*}\label{eq:phi-estimate-eta}
|\Phi_2(t;\xi+i\eta)|\leq\frac{1}{1+e^{-2\eta\tau}}\int_{-\infty}^{T_+}|q(y)-q(y+\tau)|dy\\
+\left(\frac{1-e^{-2\eta\tau}}{1+e^{-2\eta\tau}}\right)\|q\|_{\fs{L}^1}
+\frac{e^{-2\eta\tau}}{1+e^{-2\eta\tau}}\int_{t}^{t+\tau}|q(y)|dy.
\end{multline*}
Therefore, we conclude that, for fixed $\eta\geq0$, $|\Phi_2(t;\xi+i\eta)|\rightarrow0$
uniformly in $\Omega$ as $|\xi|\rightarrow\infty$.
 
Now, consider $\zeta=\xi-i\eta$ where $\eta>0$ is fixed such
that $\eta<\infty$. It is straightforward to show that
\begin{multline*}
|\Phi_2(t;\xi-i\eta)|
\leq\frac{e^{2\eta W}}{1+e^{2\eta\tau}}
\int_{-\infty}^{T_+}|q(y)-q(y+\tau)|dy\\
+\left(\frac{1-e^{-2\eta\tau}}{1+e^{-2\eta\tau}}\right)e^{2\eta W}\|q\|_{\fs{L}^1}
+\frac{e^{2\eta W}}{1+e^{2\eta\tau}}\int_{t}^{t+\tau}|q(y)|dy.
\end{multline*}
Therefore, we can again show that $|\Phi_2(t;\xi-i\eta)|\rightarrow0$ 
uniformly for $t\in\Omega$ as $|\xi|\rightarrow\infty$.

Finally, combining the results obtained above, we have 
\[
\lim_{|\xi|\rightarrow\infty}\|\Phi_2(t;\xi+i\eta)\|_{\fs{L}^{\infty}(\Omega)}=0,
\] 
for fixed $\eta\in\field{R}$ and $\eta<\infty$.

\section{Proof of Lemma~\ref{lemma:inverse-FL}}\label{app:inverse-FL}

In the following, we assume that the functions $F$ and $G$ are redefined appropriately 
so that factors of the form $e^{-2i\zeta T}$ are no longer present. Therefore,
we prove the lemma for the case $T=0$.

Observing that $F(\xi+i\eta)$, as a function of $\xi\in\field{R}$, is in
$\fs{L}^2$ for $\eta\geq0$, it is evident that
$f(\tau)\in\fs{L}^2$ and it is supported in $\Omega=[0,\infty)$. Using 
Cauchy's estimate taken on the disc
$\{\zeta'\in\field{C}_+:|\zeta'-(\xi+i\eta)|\leq\eta/2,\,\eta>0\}$, we have
\[
|F'(\xi+i\eta)|\leq\frac{2C}{\eta(1+|\xi+i\eta|-\eta/2)},\quad\eta>0,
\]
therefore, $F'(\xi+i\eta)$, as a function of $\xi\in\field{R}$, is in
$\fs{L}^2$ for $\eta>0$. From Paley-Wiener theorem, we know
that the boundary function $\lim_{\eta\rightarrow0} F'(\xi+i\eta)$ exists 
(which coincides with $F'(\xi)$ in this case) and it belongs to $\fs{L}^2$. From 
Cauchy-Schwartz inequality~\cite[Chap.~10]{Jones2001}, we have
\begin{align*}
    \biggl|\int_{\Omega}\sqrt{1+\tau^2}f(\tau)\cdot\frac{d\tau}{\sqrt{1+\tau^2}}\biggl|^2
&\leq\frac{\pi}{2}\int_{\Omega}(1+\tau^2)|f(\tau)|^2d\tau,
\end{align*}
so that $\|f\|_{1}\leq({1}/{4})(\|F(\xi)\|_2^2+\|F'(\xi)\|_2^2)$. The last step 
follows from Plancherel's theorem~\cite[Chap.~13]{Jones2001}. Therefore, we conclude
that $f(\tau)\in\fs{L}^1\cap\fs{L}^2$. 
    
Next, under the assumption of part (b), we would like to show that $f(\tau)$ is absolutely
continuous in $[0,\infty)$, i.e., there
    exists a complex-valued function $f^{(1)}(\tau)$ such that 
\begin{equation}\label{eq:absolute-continuity}
    f(\tau)=f(0+)+\int_0^{\tau}f^{(1)}(\tau')d\tau',\quad \tau\in\Omega.
\end{equation}
Let us observe the property
\begin{equation*}
\int^{\infty}_0\left[\int_0^{\tau}f^{(1)}(\tau')d\tau'\right]e^{i\zeta \tau}d\tau
=-\frac{1}{i\zeta}\int^{\infty}_0f^{(1)}(\tau)e^{i\zeta \tau}d\tau.
\end{equation*}
Note that $g(\tau)\in\fs{L}^1\cap\fs{L}^2$ since
$G(\zeta)$ satisfies the same kind of estimate as that of $F(\zeta)$. The
relation~\eqref{eq:absolute-continuity} then follows by noting that
$G(\zeta)/(-i\zeta)$ is the Fourier-Laplace transform of
$\int_0^{\tau}g(\tau')d\tau'$ and ${G(\zeta)}/{i\zeta}=F(\zeta)-{\mu}/{i\zeta}$;
therefore, $f^{(1)}(\tau)=g(\tau)$. Absolute continuity 
of $f(\tau)$ also implies continuity of $f(\tau)$ over $[0,\infty)$; consequently, the limit $f(0-)$
exists and equals $\mu$. Finally, using the Lebesgue's 
theorem~\cite[Chap.~16]{Jones2001} on differentiation, we
have $\partial_{\tau}f(\tau)=g(\tau)$ almost everywhere. It also  follows that 
$|f(\tau)|\leq |\mu|+\|g\|_{1}$, i.e., $f(\tau)$ is
bounded on $[0,\infty)$. 


\providecommand{\noopsort}[1]{}\providecommand{\singleletter}[1]{#1}%
\end{document}